\documentclass[a4paper,onecolumn,unpublished]{quantumarticle}
\pdfoutput=1
\PassOptionsToPackage{hyphens}{url}
\usepackage[dvipsnames]{xcolor}
\usepackage{graphicx}
\usepackage[ colorlinks = true, 
             linkcolor = violet,
             urlcolor  = ForestGreen,
             citecolor = ForestGreen,
             anchorcolor = Violet,
]{hyperref} 
\usepackage{amsmath,amssymb,amsthm,mathtools,mathrsfs}
\usepackage{braket}
\usepackage{pifont}
\usepackage{authblk} 
\usepackage{bbm}
\usepackage{tikz}
\usepackage[compat=1.1.0]{tikz-feynman}
\usepackage{dsfont}
\usepackage{enumerate}
\usepackage{tocbasic}
\usepackage{physics}
\usepackage[all]{xy}
\usepackage{cite}

\DeclareTOCStyleEntry[
  beforeskip=0.3em plus 1pt,
  pagenumberformat=\textbf
]{tocline}{section}

\makeatletter
\renewcommand*\env@matrix[1][\arraystretch]{%
  \edef\arraystretch{#1}%
  \hskip -\arraycolsep
  \let\@ifnextchar\new@ifnextchar
  \array{*\c@MaxMatrixCols c}}
\makeatother

\theoremstyle{plain}
\newtheorem{theorem}[equation]{Theorem}
\newtheorem{lemma}[equation]{Lemma}
\newtheorem{proposition}[equation]{Proposition}
\newtheorem{corollary}[equation]{Corollary}
\theoremstyle{definition}
\newtheorem{definition}[equation]{Definition}

\newtheorem{question}[equation]{Question}

\newtheorem{problem}[equation]{Problem}
\newtheorem{example}[equation]{Example}
\newtheorem{exercise}[equation]{Exercise}
\newtheorem*{answer}{Answer}
\newtheorem*{solution}{Solution}
\newtheorem{remark}[equation]{Remark}
\newtheorem{terminology}[equation]{Terminology}

\newtheorem{notation}[equation]{Notation}
\newtheorem{noterm}[equation]{Notation and Terminology}

\numberwithin{equation}{section}

\newcommand{\be}{\begin{equation}}
\newcommand{\ee}{\end{equation}}
\newcommand{\bx}{\begin{example}}
\newcommand{\ex}{\end{example}}
\newcommand{\bex}{\begin{exercise}}
\newcommand{\eex}{\end{exercise}}
\newcommand{\ban}{\begin{answer}}
\newcommand{\ean}{\end{answer}}
\newcommand{\bt}{\begin{theorem}}
\newcommand{\et}{\end{theorem}}
\newcommand{\bc}{\begin{corollary}}
\newcommand{\ec}{\end{corollary}}
\newcommand{\blem}{\begin{lemma}}
\newcommand{\elem}{\end{lemma}}
\newcommand{\bp}{\begin{problem}}
\newcommand{\ep}{\end{problem}}
\newcommand{\bn}{\begin{proposition}}
\newcommand{\en}{\end{proposition}}
\newcommand{\bd}{\begin{definition}}
\newcommand{\ed}{\end{definition}}
\newcommand{\bq}{\begin{question}}
\newcommand{\eq}{\end{question}}
\newcommand{\bprf}{\begin{proof}}
\newcommand{\eprf}{\end{proof}}
\newcommand{\br}{\begin{remark}}
\newcommand{\er}{\end{remark}}
\newcommand{\bs}{\begin{solution}}
\newcommand{\es}{\end{solution}}
\newcommand{\bnt}{\begin{noterm}}
\newcommand{\ent}{\end{noterm}}
\newcommand{\bnot}{\begin{notation}}
\newcommand{\enot}{\end{notation}}
\newcommand{\bterm}{\begin{terminology}}
\newcommand{\eterm}{\end{terminology}}

\newcommand{\id}{\mathrm{id}}

\def\R{{{\mathbb R}}}
\def\C{{{\mathbb C}}}

\def\Q{{{\mathbb Q}}}

\DeclareMathAlphabet{\mathpzc}{OT1}{pzc}{m}{it}
 \DeclareFontFamily{OT1}{pzc}{}
 \DeclareFontShape{OT1}{pzc}{m}{it}{ <-> s*[1.2] pzcmi7t }{}
 \DeclareMathAlphabet{\mathpzc}{OT1}{pzc}{m}{it}

\def\A{\mathcal{A}}

\def\E{\mathcal{E}}
\def\H{\mathcal{H}}
\def\F{\mathcal{F}}

\def\J{\mathscr{J}}
\def\C{\mathbb{C}}
\def\Q{\mathcal{Q}}
\def\R{\mathbb{R}}

\title{On the entropy of a pseudo-density matrix}
\author{James Fullwood}
\email{fullwood@hainanu.edu.cn}
\affiliation{School of Mathematics and Statistics, Hainan University, Haikou, Hainan, 570228, China}
\affiliation{Hainan International Exchange Center for Theoretical Physics, Haikou, Hainan, 570228, China}
\author{Boyu Yang}
\affiliation{School of Mathematics and Statistics, Hainan University, Haikou, Hainan, 570228, China}
\newcommand{\Addresses}{{
}}

\begin{document}
\emergencystretch 2em

\maketitle

\vspace{-7mm}
\tableofcontents

\begin{abstract}
The pseudo-density matrix (PDM) formalism naturally extends the notion of density operator to the spatiotemporal domain. While PDMs are Hermitian and of unit trace, they admit negative eigenvalues when encoding temporal correlations unattainable for spacelike separated systems. Consequently, there have been various approaches to extending the von~Neumann entropy—which is only defined for positive operators—to PDMs. Here, we prove that there exists a unique extension of the von~Neumann entropy functional to Hermitian matrices of unit trace satisfying two simple assumptions: unitary invariance and strong additivity with respect to affine combinations within the interval $[-1,1]$, which we prove contains the eigenvalues of a PDM. We also prove that this unique extension of von~Neumann entropy to PDMs is subadditive for single qubit dynamics, and we analyze its behavior in a number of examples.
\end{abstract}

\section{Introduction}

The axiomatic approach to the formulation of information measures goes back to the seminal work of Shannon~\cite{Shannon}, who proved that there is a unique functional $H$ on finite probability distributions satisfying three fundamental properties: Continuity, monotonicity with respect to the number of inputs for uniform distributions, and strong additivity with respect to convex combinations. For a finite probability vector $(p_1,\ldots,p_n)$, the functional $H$ is given by
\[
H(p_1,\ldots,p_n)=-\sum_{i=1}^{n}p_i\log(p_i)\, ,
\] 
and is now referred to as the \emph{Shannon entropy} of $(p_1,\ldots,p_n)$. Interestingly, while Shannon established this rigorous mathematical foundation for classical information theory, the analogous quantum measure—the von Neumann entropy—predated it by over two decades. Introduced by von Neumann in 1927~\cite{vonNeumann_1927}, the quantum entropy $S(\rho)=-\tr(\rho\log(\rho))$ of a density operator $\rho$ was originally derived not from abstract information-theoretic axioms, but through a thermodynamic thought experiment involving the reversible separation of a quantum gas using macroscopic semi-permeable membranes. It was not until much later that mathematical physicists, notably Ochs and Wehrl in the 1970s~\cite{Ochs_1975,Wehrl_1978}, established strict axiomatic characterizations of the von~Neumann entropy. The axiomatic approach to the study of information measures continues to evolve into the present day~\cite{Aczel,Faddeev,Furuichi,Renyi,Pe92,Ebanks,Baez_2011,BaFr14,leinster2019short,FuPa21,Fullwood_AXM,PaEntropy}, including the functorial characterization of Shannon entropy by Baez, Fritz, and Leinster in terms of information loss~\cite{Baez_2011}, a framework which Parzygnat recently extended to the quantum domain to provide a functorial characterization of von~Neumann entropy~\cite{PaEntropy}.

In this work, we apply the axiomatic approach to the extension of von~Neumann entropy to \emph{pseudo-density matrices} (PDMs), which provide a natural extension of the density operator formalism into the spatiotemporal domain. In particular, PDMs encode temporal correlations associated with sequential measurements on a quantum system in precisely the same manner that a density matrix encodes correlations between spacelike separated systems. Moreover, PDMs also encode the \emph{dynamics} between sequential measurements, whether it be unitary evolution or dissipative dynamics corresponding to system-environment interactions. While PDMs are Hermitian and of unit trace, they admit negative eigenvalues when they encode temporal correlations unattainable by spacelike separated measurements. Since their introduction by Fitzsimons, Jones and Vedral in 2015~\cite{FJV15}, PDMs have provided insights into various spatiotemporal aspects of quantum information, including channel capacity~\cite{Pisar19}, teleportation in time~\cite{MVVAPGDG21}, the black hole information problem~\cite{Marletto_2020}, quantum Bayes' rules~\cite{FuPa22a,FuPa25}, quantum causal inference~\cite{Liu_2024,Liu_2025}, and testing macroscopic realism~\cite{FullwoodWu_2025,Comar_2026}. The PDM formalism is also supported by a firm axiomatic foundation~\cite{LiNg23,LieFu_2025}, and it has recently been shown that PDMs are equivalent to the notion of canonical quantum state over time~\cite{HHPBS17,FuPa24}.

While the negative eigenvalues of a PDM serve as a witness to quantum correlations which imply causation~\cite{FJV15,LiuV_2025}, a non-positive spectrum renders the von~Neumann entropy ill-defined, as the logarithm is only defined for positive operators. Moreover, as even non-Hermitian matrices naturally arise in the context of quantum states over time~\cite{Lie_2026}, spacetime states~\cite{Guo_2025,Milekhin2025,Diaz_2026}, local-density operators~\cite{FullwoodYang_25}, non-Hermitian Hamiltonians~\cite{Ashida_2020}, post-selection~\cite{Aharonov_1988,AHLLBL20}, and the two-vector formalism~\cite{AhVa08}, the need to evaluate entropy for non-positive or complex spectra has become increasingly relevant. Consequently, there have been various proposals for extending von~Neumann entropy to PDMs and non-Hermitian matrices. The recent work Ref.~\cite{Chen_2026} compares four distinct extensions: pseudo-entropy~\cite{Nakata_2020}, modified pseudo-entropy~\cite{JSK23}, SVD entropy~\cite{PaTK23}, and ABB entropy~\cite{ABB_2000}. Although pseudo-entropy accommodates post-selection, it frequently yields complex values or severe divergences. SVD entropy circumvents this by utilizing singular values to guarantee a real, bounded measure. Meanwhile, ABB entropy admits an explicit operational interpretation as an entanglement distillation monotone.

Distinct from these approaches, an extension of von~Neumann entropy to normal operators---referred to as `FP entropy' in Ref.~\cite{PaTK23}---has been utilized to evaluate negative conditional entropy in post-selected states~\cite{SSW14}, to investigate negative central charges in non-Hermitian systems~\cite{TTC22}, and to compute the Page curve for non-Hermitian systems whose eigenstates exhibit many-body quantum chaos~\cite{CiKu23}. Furthermore, such a notion of entropy has been shown to provide a natural notion of entropy for pseudo-density matrices~\cite{FuPa23}, and is the foundational measure used to define quantum mutual information in time in Ref.~\cite{FullwoodWu_2025a}. However, despite its utility across diverse contexts, this extension of von~Neumann entropy has thus far lacked a rigorous mathematical justification, a gap we resolve in this work by providing a strict axiomatic characterization. In particular, we prove that such an extension of von~Neumann entropy to unit trace Hermitian matrices is singled out by two simple assumptions: unitary invariance and strong additivity with respect to affine combinations within the interval $[-1,1]$. The relaxation of strong additivity with respect to convex combinations to allow for affine combinations within the interval $[-1,1]$ is motivated by our proof that the eigenvalues of a pseudo-density matrix never exceed $1$ in absolute value (Theorem~\ref{PDMSPXCTRA27}). We also prove that the entropy of a pseudo-density matrix is subadditive for single qubit dynamics, which resolves a special case of a conjecture put forth in Ref.~\cite{FuPa23}. Finally, we consider some fundamental examples of open system dynamics---such as depolarizing, amplitude damping, bit-flip and phase-flip channels---and we establish bounds on the entropy of the associated pseudo-density matrices.   

\section{Pseudo-density matrices and their spectra}

While multipartite density matrices encode correlations between measurements performed on spacelike separated quantum systems, multipartite pseudo-density matrices encode correlations between sequential measurements performed on the same quantum system. To motivate the definition of a pseudo-density matrix, let $A$ and $B$ denote quantum systems each consisting of a single qubit, so that we may identify the associated Hilbert spaces $\H_A$ and $\H_B$ with $\C^2$. If $A$ and $B$ are spacelike separated, then the joint state of the composite system $AB$ is represented by a density matrix $\rho_{AB}$ acting on $\C^2\otimes \C^2$, which is a positive operator of unit trace. Expanding $\rho_{AB}$ with respect to the Pauli basis yields
\be \label{2QBTDX87}
\rho_{AB}=\frac{1}{4}\sum_{i=0}^{3}c_{ij}\sigma_i\otimes \sigma_j\, ,
\ee
which by properties of Pauli matrices implies 
\[
c_{ij}=\tr(\rho_{AB}\,\sigma_i\otimes \sigma_j)\equiv \langle \sigma_i\otimes \sigma_j \rangle\, ,
\]
where $\langle \sigma_i\otimes \sigma_j \rangle$ is the expectation value of the product of local measurements of $\sigma_i$ and $\sigma_j$ performed on systems $A$ and $B$, respectively. All of this is standard. 

Now suppose that $A$ and $B$ are instead timelike separated, and correspond to the \emph{same} qubit at two times $t_A$ and $t_B>t_A$. Furthermore, suppose that this single qubit was prepared in state $\rho_A$ prior to a sequential measurement of $\sigma_i$ followed by $\sigma_j$ at times $t_A$ and $t_B$ respectively. If the qubit evolves according to a quantum channel $\E$ between measurements, then it follows from the L\"{u}ders-von~Neumann projection postulate that the 2-time expectation value $\langle \sigma_i\, ,\sigma_j \rangle$ associated with such a sequential measurement scenario is given by 
\[
\langle \sigma_i\, ,\sigma_j \rangle=\tr(\E(\Pi_i^+\rho_A \Pi_i^+)\sigma_j)-\tr(\E(\Pi_i^-\rho_A \Pi_i^-)\sigma_j)\, ,
\]
where $\Pi_i^{\pm}$ and $\Pi_j^{\pm}$ are the projection operators such that $\sigma_i=\Pi_i^+-\Pi_i^-$ and $\sigma_j=\Pi_j^+-\Pi_j^-$. Replacing $c_{ij}$ in the right-hand side (RHS) of Eq.~\eqref{2QBTDX87} with the 2-time expectation values $\langle \sigma_i\, ,\sigma_j \rangle$ then yields the pseudo-density matrix $\varrho_{AB}$ associated with the timelike separated qubits $A$ and $B$, so that
\be \label{PDM71}
\varrho_{AB}=\frac{1}{4}\sum_{i=0}^{3}\langle \sigma_i\, ,\sigma_j \rangle\sigma_i\otimes \sigma_j \,.
\ee
While the pseudo-density matrix $\varrho_{AB}$ is Hermitian and of unit trace, the expectation values $\langle \sigma_i\, ,\sigma_j \rangle$ may endow $\varrho_{AB}$ with negative eigenvalues, so that $\varrho_{AB}$ is not a valid density matrix in general. Nevertheless, $\rho_A=\tr_B(\varrho_{AB})$ and $\rho_B=\tr_A(\varrho_{AB})$ are both valid density operators, and it follows from properties of the Pauli matrices that for all $i$ and $j$ we have
\[
\langle \sigma_i\, ,\sigma_j \rangle=\tr(\varrho_{AB}\, \sigma_i\otimes \sigma_j)\, .
\]
It then follows that the pseudo-density matrix $\varrho_{AB}$ encodes the 2-time expectation values $\langle \sigma_i\, ,\sigma_j \rangle$ in precisely the same manner that the density matrix $\rho_{AB}$ encodes the spatial expectation values $\langle \sigma_i\otimes \sigma_j \rangle$. As such, the pseudo-density matrix $\varrho_{AB}$ yields a natural notion of quantum state for timelike separated qubits. For a concrete example, if $\rho=\dyad{0}{0}$ and $\E$ is the identity channel, then 
\be \label{BCSTX97}
\varrho_{AB}=
\begin{pmatrix}
1&0&0&0\\
0&0&\frac{1}{2}&0\\
0&\frac{1}{2}&0&0\\
0&0&0&0\\
\end{pmatrix}\
\, ,
\ee
which has eigenvalues $(1,0,1/2,-1/2)$. The non-positivity of $\varrho_{AB}$ then implies that the correlations encoded by $\varrho_{AB}$ may not be realized by a pair of spacelike separated qubits. 

In Ref.~\cite{HHPBS17}, it was shown that the 2-time pseudo-density matrix $\varrho_{AB}$ as given by \eqref{PDM71} may be given by the compact formula
\be \label{GBLXFXM627}
\varrho_{AB}=\frac{1}{2}\big\{\rho_A\otimes \mathds{1}_B\, , \mathscr{J}[\E]\big\}\, ,
\ee 
where $\{\cdot\, ,\cdot\}$ denotes the anticommutator, and $\J[\E]=\sum_{i,j}\dyad{i}{j}\otimes \E(\dyad{j}{i})$ is the \emph{Jamio\l kowski} matrix\footnote{The Jamio\l kowski matrix $\J[\E]$ is not to be confused with the Choi matrix of the channel $\E$, which is obtained from $\J[\E]$ by partial transposition.} of the quantum channel $\E$, which we recall is a completely positive, trace-preserving (CPTP) linear map governing the dynamics between measurements in the associated 2-point sequential measurement scenario. As the RHS of Eq.~\eqref{GBLXFXM627} is well-defined for any two (finite-dimensional) quantum systems $A$ and $B$ corresponding to the input and output of a quantum channel $\E$, we take the RHS of Eq.~\eqref{GBLXFXM627} as the general definition of a pseudo-density matrix associated with any state-channel pair $(\rho_A,\E)$. Moreover, it was shown in Ref.~\cite{FuPa24} that the pseudo-density matrix $\varrho_{AB}$ as given by  Eq.~\eqref{GBLXFXM627} retains the operational interpretation of the 2-qubit pseudo-density matrix in terms of encoding 2-time expectation values.  

The 2-time pseudo-density matrix $\varrho_{AB}$ associated with the state-channel pair $(\rho_A,\E)$ as given by Eq.~\eqref{GBLXFXM627} admits a natural extension to $n$-point sequential measurement scenarios~\cite{Fullwood_2025a,Liu_2025}. In particular, suppose a quantum system initially in state $\rho$ is measured at $n$ times $t_1<t_2<\cdots<t_n$, and suppose that the system evolves according to a quantum channel $\E_i$ between the measurements performed at times $t_i$ and $t_{i+1}$ for $i=1,\ldots,n-1$. We then define the associated $n$-time pseudo-density matrix $\varrho_{A_1\cdots A_n}$ by the recursive formula
\be \label{NTXPDM679}
\varrho_{A_1\cdots A_n}=\frac{1}{2}\big\{\varrho_{A_1\cdots A_{n-1}}\otimes \mathds{1}_{A_n},\J[\E_{n-1}\circ \tr_{1\cdots (n-2)}] \big\}\, ,
\ee
where $\tr_{1\cdots (n-2)}$ denotes the partial trace over the subsystems $A_1\cdots A_{n-2}$.

As we saw with the simple example of a 2-time pseudo-density matrix given by Eq.~\eqref{BCSTX97}, all of its eigenvalues were contained in the interval $[-1,1]$. The following result establishes that the same bound holds for the eigenvalues of an $n$-time pseudo-density matrix as given by Eq.~\eqref{NTXPDM679}. Since pseudo-density matrices are Hermitian and of unit trace, this implies that the spectral decomposition of an $n$-time pseudo-density matrix $\varrho_{A_1\cdots A_n}$ takes the form of an affine combination
\[
\varrho_{A_1\cdots A_n}=\sum_i \lambda_i \dyad{\psi_i}{\psi_i}\, ,
\]
where $\{\lambda_i\}\subset [-1,1]$ is the multi-set of eigenvalues of $\varrho_{A_1\cdots A_n}$ (so that $\sum_i \lambda_i = 1$), and $\{\ket{\psi_i}\}$ is an orthonormal basis of $\H_{A_1}\otimes \cdots \otimes \H_{A_n}$. This structural property will be crucial for our axiomatic derivation of the entropy of a pseudo-density matrix in the next section.

\bt \label{PDMSPXCTRA27}
The spectrum of the pseudo-density matrix $\varrho_{A_1\cdots A_n}$ as given by \eqref{NTXPDM679} is contained in the interval $[-1,1]$.
\et 

Before giving a proof we first prove two lemata. For this, recall that a quantum channel from system $A$ to $B$ consists of a CPTP map $\E:\mathcal{L}(\H_A)\to \mathcal{L}(\H_B)$, where $\mathcal{L}(\H)$ denotes the algebra of linear operators on a Hilbert space $\H$.

\blem \label{lem:jamiolkowski-norm}
Let $A$ and $B$ denote finite-dimensional quantum systems, and let $\E: \mathcal{L}(\H_A) \to \mathcal{L}(\H_B)$ be a quantum channel. Then the Jamio\l kowski matrix $\J[\E]=\sum_{i,j}\dyad{i}{j}\otimes\E(\dyad{j}{i})$ satisfies $\norm{\J[\E]}_{\infty} \le 1$.
\elem

\bprf
Suppose the Kraus representation of $\E$ is given by $\E(\rho) = \sum_{\alpha} K_{\alpha} \rho K_{\alpha}^\dagger$, where the trace-preserving condition mandates $\sum_{\alpha} K_{\alpha}^\dagger K_{\alpha} = \mathds{1}_A$. Using the swap operator $S_A = \sum_{i,j} \dyad{i}{j} \otimes \dyad{j}{i}$ on $\H_A \otimes \H_A$, we can rewrite the Jamio\l kowski matrix as
\[
\J[\E] = \sum_{\alpha} (\mathds{1}_A \otimes K_{\alpha}) S_A (\mathds{1}_A \otimes K_{\alpha}^\dagger)\, .
\]
Now consider an arbitrary normalized vector $\ket{\psi} \in \H_A \otimes \H_B$. We can always express $\ket{\psi}$ as $\ket{\psi} = (V \otimes \mathds{1}_B) \ket{\Phi^+_B}$, where $\ket{\Phi^+_B} = \sum_{k=1}^{d_B} \ket{k}_B \ket{k}_B$ is the unnormalized maximally entangled state on $\H_B \otimes \H_B$, and $V: \H_B \to \H_A$ is a linear operator satisfying $\tr(V^\dagger V) = \braket{\psi}{\psi} = 1$. Evaluating the expectation value of $\J[\E]$ with respect to $\ket{\psi}$ yields
\begin{align*}
\bra{\psi} \J[\E] \ket{\psi} &= \sum_{\alpha} \bra{\Phi^+_B} (V^\dagger \otimes \mathds{1}_B) (\mathds{1}_A \otimes K_{\alpha}) S_A (\mathds{1}_A \otimes K_{\alpha}^\dagger) (V \otimes \mathds{1}_B) \ket{\Phi^+_B} \\
&= \sum_{\alpha} \bra{\Phi^+_B} (V^\dagger \otimes K_{\alpha}) S_A (V \otimes K_{\alpha}^\dagger) \ket{\Phi^+_B}\, .
\end{align*}
Now since 
\[
S_A (V \otimes K_{\alpha}^\dagger) \ket{\Phi^+_B} = (K_{\alpha}^\dagger \otimes V) \ket{\Phi^+_B}\, ,
\]
the inner product simplifies to
\[
\bra{\psi} \J[\E] \ket{\psi} = \sum_{\alpha} \bra{\Phi^+_B} (V^\dagger K_{\alpha}^\dagger \otimes K_{\alpha} V) \ket{\Phi^+_B}\, .
\]
Moreover, note that for any operators $X$ and  $Y$ acting on $\H_B$ we have the standard identity 
\[
\bra{\Phi^+_B} (X \otimes Y) \ket{\Phi^+_B} = \sum_{k,l} \bra{k} X \ket{l} \bra{k} Y \ket{l}\, ,
\]
thus 
\[
\bra{\Phi^+_B} (B_{\alpha}^\dagger \otimes B_{\alpha}) \ket{\Phi^+_B} = \sum_{k,l} \bra{k} B_\alpha^\dagger \ket{l} \bra{k} B_\alpha \ket{l} = \sum_{k,l} \overline{\bra{l} B_\alpha \ket{k}} \bra{k} B_\alpha \ket{l}\, ,
\]
where $B_{\alpha} = K_{\alpha} V$. Applying the Cauchy-Schwarz inequality to the sum over matrix elements then yields
\[
\left| \sum_{k,l} \overline{\bra{l} B_{\alpha} \ket{k}} \bra{k} B_{\alpha} \ket{l} \right| \le \sqrt{ \sum_{k,l} |\bra{l} B_{\alpha} \ket{k}|^2 } \sqrt{ \sum_{k,l} |\bra{k} B_{\alpha} \ket{l}|^2 } = \sum_{k,l} |\bra{k} B_{\alpha} \ket{l}|^2 = \tr(B_{\alpha}^\dagger B_{\alpha})\, .
\]
Taking the absolute value of the total expectation value and summing over $\alpha$ we obtain
\[
|\bra{\psi} \J[\E] \ket{\psi}| \le \sum_{\alpha} \tr(B_{\alpha}^\dagger B_{\alpha}) = \sum_{\alpha} \tr(V^\dagger K_{\alpha}^\dagger K_{\alpha} V) = \tr\left( V^\dagger \left( \sum_{\alpha} K_{\alpha}^\dagger K_{\alpha} \right) V \right)\, .
\]
Furthermore, since $\sum_{\alpha} K_{\alpha}^\dagger K_{\alpha} = \mathds{1}_A$, the above sum collapses to
\[
|\bra{\psi} \J[\E] \ket{\psi}| \le \tr(V^\dagger \mathds{1}_A V) = \tr(V^\dagger V) = 1\, .
\]
Since $\J[\E]$ is Hermitian,
\[
\norm{\J[\E]}_\infty=\sup_{\norm{\psi}=1}
\left|\bra{\psi}\J[\E]\ket{\psi}\right|\, ,
\]
thus preceding bound implies $\norm{\J[\E]}_\infty\le1$, as desired.
\eprf

\blem \label{lem:anticommutator-norm}
Let $A$ and $B$ denote finite-dimensional quantum systems, and suppose $X \in \mathcal{L}(\H_A)$ and $Y \in \mathcal{L}(\H_A \otimes \H_B)$ are such that $\norm{X}_{\infty} \le 1$ and $\norm{Y}_{\infty} \le 1$. Then
\[
\norm{\frac{1}{2}\{X \otimes \mathds{1}_B, Y\}}_{\infty} \le 1\, .
\]
\elem

\bprf
Let $W = X \otimes \mathds{1}_B$. Because the operator norm of a tensor product of operators is the product of their respective operator norms, we have $\norm{W}_{\infty} = \norm{X}_{\infty} \norm{\mathds{1}_B}_{\infty} = \norm{X}_{\infty} \le 1$. By the triangle inequality and the submultiplicativity of the operator norm, we have
\begin{align*}
\norm{\frac{1}{2}\{W, Y\}}_{\infty} &= \frac{1}{2}\norm{WY + YW}_{\infty} \le \frac{1}{2} \big( \norm{WY}_{\infty} + \norm{YW}_{\infty} \big) \\
&\le \frac{1}{2} \big( \norm{W}_{\infty}\norm{Y}_{\infty} + \norm{Y}_{\infty}\norm{W}_{\infty} \big)= \norm{W}_{\infty}\norm{Y}_{\infty}\, .
\end{align*}
Since $\norm{W}_{\infty} \le 1$ and $\norm{Y}_{\infty} \le 1$, it immediately follows that $\norm{\frac{1}{2}\{X \otimes \mathds{1}_B, Y\}}_{\infty} \le 1$, as  desired.
\eprf

\bprf[Proof of Theorem~\ref{PDMSPXCTRA27}]
We proceed by induction on $n$. For the base case $n=2$, Eq.~\eqref{GBLXFXM627} defines the 2-time pseudo-density matrix as $\varrho_{AB}=\frac{1}{2}\{\rho_A\otimes\mathds{1}_B,\J[\E]\}$. 
Because $\rho_A$ is a valid density operator, its spectrum lies in $[0,1]$, ensuring $\norm{\rho_A}_{\infty} \le 1$. By Lemma~\ref{lem:jamiolkowski-norm}, the Jamio\l kowski matrix satisfies $\norm{\J[\E]}_{\infty} \le 1$. 
Applying Lemma~\ref{lem:anticommutator-norm} with $X = \rho_A$ and $Y = \J[\E]$, we immediately obtain $\norm{\varrho_{AB}}_{\infty} \le 1$. 
Since $\varrho_{AB}$ is a Hermitian operator, its eigenvalues are bounded in absolute value by its operator norm, forcing the spectrum of $\varrho_{AB}$ to lie strictly within the interval $[-1, 1]$. For the inductive step, assume the claim holds for $n-1$, so that the spectrum of $\varrho_{A_1\cdots A_{n-1}}$ is contained in $[-1, 1]$. As $\varrho_{A_1\cdots A_{n-1}}$ is Hermitian, this implies $\norm{\varrho_{A_1\cdots A_{n-1}}}_{\infty} \le 1$. Now the recursive formula for the $n$-time PDM is given in Eq.~\eqref{NTXPDM679} as $\varrho_{A_1\cdots A_n}=\frac{1}{2}\{\varrho_{A_1\cdots A_{n-1}}\otimes \mathds{1}_{A_n},\J[\Phi]\}$, where $\Phi = \E_{n-1} \circ \tr_{1\cdots(n-2)}$. As the partial trace $\tr_{1\cdots(n-2)}$ and $\E_{n-1}$ are both CPTP, their composition $\Phi$ is also a CPTP map, hence $\norm{\J[\Phi]}_{\infty} \le 1$ by Lemma~\ref{lem:jamiolkowski-norm}. Applying Lemma~\ref{lem:anticommutator-norm} with $X = \varrho_{A_1\cdots A_{n-1}}$ and $Y = \J[\Phi]$, we find that $\norm{\varrho_{A_1\cdots A_n}}_{\infty} \le 1$. As the $n$-time PDM $\varrho_{A_1\cdots A_n}$ is Hermitian, its spectrum must be contained in the interval $[-1, 1]$, as desired.
\eprf

\section{Axiomatic derivation of PDM entropy}

In this section, we prove that two natural properties isolate a unique extension of von~Neumann entropy to PDMs: unitary invariance and strong additivity with respect to affine combinations within the interval $[-1,1]$. As von~Neumann entropy satisfies strong additivity for \emph{convex} combinations, our axiomatic derivation of entropy for pseudo-density matrices merely relaxes this condition to accommodate affine combinations. Such a relaxation is certainly natural in light of Theorem~\ref{PDMSPXCTRA27}, which establishes that the spectrum of a pseudo-density matrix is contained in $[-1,1]$. To fix notation for the statement of our result, we let $\mathbb{H}_{n}^{1}$ denote the set of all $n\times n$ Hermitian matrices of unit trace, and we let $\mathbb{H}^1=\bigsqcup_{n=1}^\infty \mathbb{H}_n^{1}$. For a Hermitian matrix $X$, the expression $X\log|X|$ is understood via the functional calculus for the function $g:\R\to\R$ given by $g(t)=t\log|t|$ for $t\neq0$ and $g(0)=0$.

\bt \label{FPXCXH747}
Suppose $S:\mathbb{H}^{1}\to \R$ is a map on the set of all unit trace Hermitian matrices satisfying the following conditions. 
\begin{enumerate}
\item \label{FPXCXH7471}
\textbf{Reduction to von~Neumann entropy.} $S(\rho)=-\tr\big(\rho\log(\rho)\big)$ for every positive element $\rho\in \mathbb{H}^1$.
\item \label{FPXCXH7472}
\textbf{Unitary invariance.} $S(UXU^\dagger)=S(X)$ for every $X\in \mathbb{H}^1$ and for every unitary $U$ which is the same size as $X$.
\item \label{FPXCXH7473}
\textbf{Strong additivity.}
If $\{X_i\}_{i=1}^{k}$ is such that $X_i\in\mathbb{H}_{n_i}^1$, and if $(p_1,\ldots,p_k)\in [-1,1]^k$ is such that $\sum_i p_i=1$, then
\be \label{SADXY347}
S\left(\bigoplus_{i=1}^{k}p_iX_i\right)= S\big(\operatorname{diag}(p_1,\ldots , p_k)\big)+\sum_{i=1}^k p_iS(X_i)\, ,
\ee
where $\bigoplus_{i=1}^k p_iX_i$ is the block-diagonal matrix with blocks $p_iX_i$, and $\operatorname{diag}(p_1,\ldots , p_k)$ is the diagonal matrix with $(p_1,\ldots,p_k)$ along the diagonal.
\end{enumerate}
Then
\be \label{FPX91}
S(X)=-\tr(X\log|X|)\qquad \forall X\in \mathbb{H}^1.
\ee
Conversely, the map $S:\mathbb{H}^{1}\to \R$ defined by \eqref{FPX91} satisfies conditions~\ref{FPXCXH7471}-\ref{FPXCXH7473}. 
\et

Before giving a proof we need to prove some preliminary results. We first set some notation. For all $n>0$, let $\mathcal{A}_n\subset \R^n$ be the subset consisting of all vectors whose entries sum to unity, let $\mathcal{Q}_n=\mathcal{A}_n\cap [-1,1]^n$, let $\Delta_n=\mathcal{A}_n\cap [0,1]^n$. We then let $\A=\bigcup_{n=1}^{\infty}\mathcal{A}_n$, $\Q=\bigcup_{n=1}^{\infty}\Q_n$, and $\Delta=\bigcup_{n=1}^{\infty}\Delta_n$. Given $p=(p_1,\ldots,p_n)\in \mathcal{Q}_n$ for some $n>0$, and $q^{(i)}=(q_{i1},\ldots,q_{im_i})\in \A_{m_i}$ for all $i\in \{1,\ldots,n\}$, we let $\bigoplus_{i=1}^{n}p_iq^{(i)}$ be the element of $\A_{m_1+\cdots+m_n}$ given by
\[
\bigoplus_{i=1}^{n}p_iq^{(i)}=\big(p_1q_{11},\ldots,p_1q_{1m_1},\ldots,p_nq_{n1},\ldots,p_nq_{nm_n}\big)\, .
\]
A function $h:\mathcal{A}\to \R$ is then said to be \emph{strongly additive} if and only if 
\be \label{STNGXADX79}
h\left(\bigoplus_{i=1}^{n}p_iq^{(i)}\right)=h(p_1,\ldots,p_n)+\sum_{i=1}^{n}p_ih\big(q^{(i)}\big)\, .
\ee

\blem \label{LMXA57}
The following statements hold.
\begin{enumerate}[i.]
\item \label{LX2}
Let $h:\A\to \R$ be the function given by $h(x_1,\ldots,x_n)=-\sum_{i=1}^{n}x_i\log|x_i|$. Then $h$ is strongly additive.
\item \label{LX4}
If $h:\A\to \R$ and $h':\A\to \R$ are strongly additive, then $h-h'$ is strongly additive. 
\item \label{LX5}
Let $h:\A\to \R$ be strongly additive, and suppose that $h(p)=-\sum_{i=1}^np_i\log(p_i)$ for all $p=(p_1,\ldots,p_n)\in \Delta$. Then $h(x_1,\ldots,x_n)=-\sum_{i=1}^{n}x_i\log|x_i|$ for all $(x_1,\ldots,x_n)\in \A$.
\end{enumerate}
\elem

\bprf
Here and throughout, we use the convention $0\log(0)=0$.

\underline{Item~\ref{LX2}}: Let $(p_1,\ldots,p_n)\in \mathcal{Q}_n$, and let $q^{(i)}=(q_{i1},\ldots,q_{im_i})\in \A_{m_i}$ for all $i\in \{1,\ldots,n\}$. We then have
\begin{align*}
h\Big(\bigoplus_{i=1}^{n}p_iq^{(i)}\Big)
&=-\sum_{i=1}^{n}\sum_{j=1}^{m_i}p_iq_{ij}\log|p_iq_{ij}|=-\sum_{i=1}^{n}\sum_{j=1}^{m_i}\left(q_{ij}p_i\log|p_i|+p_iq_{ij}\log|q_{ij}|\right)\\
&=-\sum_{i=1}^{n}p_i\log|p_i|-\sum_{i=1}^{n}p_i\sum_{j=1}^{m_i}q_{ij}\log|q_{ij}|=h(p_1,\ldots,p_n)+\sum_{i=1}^{n}p_ih\big(q^{(i)}\big)\, ,
\end{align*}
where the third equality follows from the fact that $\sum_jq_{ij}=1$. Therefore $h$ is strongly additive, as desired.

\underline{Item~\ref{LX4}}: Let $h:\A\to \R$ and $h':\A\to \R$ be strongly additive, let $(p_1,\ldots,p_n)\in \mathcal{Q}_n$ for some $n>0$, and let $q^{(i)}=(q_{i1},\ldots,q_{im_i})\in \A_{m_i}$ for all $i\in \{1,\ldots,n\}$. Then
\begin{align*}
(h-h')\Big(\bigoplus_{i=1}^{n}p_iq^{(i)}\Big)&=h\Big(\bigoplus_{i=1}^{n}p_iq^{(i)}\Big)-h'\Big(\bigoplus_{i=1}^{n}p_iq^{(i)}\Big) \\
&=h(p_1,\ldots,p_n)+\sum_{i=1}^{n}p_ih\big(q^{(i)}\big)-h'(p_1,\ldots,p_n)-\sum_{i=1}^{n}p_ih'\big(q^{(i)}\big) \\
&=(h-h')(p_1,\ldots,p_n)+\sum_{i=1}^{n}p_i(h-h')\big(q^{(i)}\big)\, ,
\end{align*}
thus $h-h'$ is strongly additive, as desired.

\underline{Item~\ref{LX5}}: Let $h^0:\A\to \R$ be the function given by $h^0(x) = -\sum_{i=1}^n x_i \log|x_i|$ for every $x=(x_1,\ldots,x_n) \in \A$. By item~\ref{LX2}, $h^0$ is strongly additive. Now let $r = h - h^0$, which is also strongly additive by item~\ref{LX4}. Furthermore, since by assumption $h(p) = -\sum_{i=1}^n p_i \log p_i = h^0(p)$ for any probability vector $p \in \Delta$, we have $r(p) = 0$ for all $p \in \Delta$. We now prove that $r(x) = 0$ for all $x \in \A$, from which the statement follows. 

We first prove that $r(p)=0$ for every $p\in \Q$. First note that appending or removing a zero entry from an element of $\A$ does not affect the value of $r$. Indeed, if $x=(x_1,\ldots,x_{n-1})\in\A_{n-1}$, then strong additivity applied to the outer vector $(1,0)\in\Delta_2$ and the inner vectors $x$ and $(1)$ gives
\[
r(x_1,\ldots,x_{n-1},0)=r(1\cdot (x_1,\ldots,x_{n-1})\oplus 0\cdot (1))=r(1,0)+r(x)=r(x)\, ,
\]
where the final equality follows from the fact that $r(1,0)=0$. It then follows that without a loss of generality we may restrict our attention to elements of $\A$ whose components are all non-zero (and similarly for elements of $\Q$ and $\Delta$). Now let $p=(p_1,\ldots,p_k)\in\mathcal{Q}_k$. If $p\in\Delta$, then $r(p)=0$. Otherwise, let
\[
s=\sum_{p_i<0}|p_i|, \qquad P=\sum_{p_i>0}p_i=1+s\, ,
\]
and let $p^{\pm}$ be the probability vectors given by
\[
p^+=\left(\frac{p_i}{P}\right)_{p_i>0}, \qquad \text{and} \qquad p^-=\left(\frac{|p_i|}{s}\right)_{p_i<0}\, .
\]
Furthermore, let $\alpha=1/P=1/(1+s)$ and $\beta=1-\alpha=s/(1+s)$, and set $y=(\alpha p_1,\ldots,\alpha p_k,\beta)$. Since $y=\alpha(p_1,\ldots,p_k)\oplus \beta(1)$, it follows from strong additivity that
\[
r(y)=r(\alpha,\beta)+\alpha r(p)=\alpha r(p)\, ,
\]
where the final equality follows from the fact that $(\alpha,\beta)\in\Delta_2$. On the other hand, up to permutation, we can decompose $y$ as the block combination $y=(1)p^+\oplus(-\beta)p^-\oplus(\beta)(1)$. Since $(1,-\beta,\beta)\in\mathcal{Q}_3$, a second application of strong additivity yields
\[
r(y)=r(1,-\beta,\beta)+1\cdot r(p^+)-\beta\cdot r(p^-)+\beta\cdot r(1)=r(1,-\beta,\beta)\, ,
\]
where we used $r(p^+)=r(p^-)=r(1)=0$ since $p^{\pm}$ and $(1)$ are all probability vectors. Equating these expressions for $r(y)$ then yields $r(p)=(1+s)r\left(1,-\frac{s}{1+s},\frac{s}{1+s}\right)$. 

We now define a function $R$ given by
\[
R(s)=(1+s)r\left(1,-\frac{s}{1+s},\frac{s}{1+s}\right),\qquad s\geq0\, ,
\]
so that $r(p)=R\big(\sum_{p_i<0}|p_i|\big)$ for every $p\in \mathcal{Q}$. Note that every $s\geq0$ is realized in this way by constructing an appropriate vector $p_s \in \mathcal{Q}_{2N}$ for some $N>0$. We now show that $R$ is in fact the zero function, which together with the fact that $r(p)=R\big(\sum_{p_i<0}|p_i|\big)$ implies that $r(p)=0$ for all $p\in \Q$.

For this, let $s,t\geq0$ and $\lambda\in[0,1]$, and choose $p_s$ and $p_t$ in $\mathcal{Q}$ with total negative masses $s$ and $t$, respectively. We then have
\begin{align*}
R\big(\lambda s+(1-\lambda)t\big) &= r\big(\lambda p_s \oplus (1-\lambda)p_t\big) \\
&= r(\lambda, 1-\lambda) + \lambda r(p_s) + (1-\lambda)r(p_t) \\
&= 0 + \lambda R(s) + (1-\lambda)R(t) \\
&=\lambda R(s) + (1-\lambda)R(t)\, ,
\end{align*}
where we utilized the fact that the total negative mass of the block vector $\lambda p_s \oplus (1-\lambda)p_t$ is exactly $\lambda s + (1-\lambda)t$, and that $r(\lambda, 1-\lambda) = 0$ since $(\lambda, 1-\lambda) \in \Delta_2$. Now since $R(0)=0$ (as $s=0$ implies $p_s \in \Delta$, where $r$ vanishes), taking $s=1$ and $t=0$ in the affine relation yields $R(u)=uR(1)$ for $0\leq u\leq1$. If $u>1$, we instead take $s=u$, $t=0$, and $\lambda=1/u$ to find $R(1) = \frac{1}{u}R(u)$, which implies $R(u)=uR(1)$. Writing $K=R(1)$, we obtain the linear relationship $R(s)=Ks$ for all $s\geq0$.

To show $K=0$, we evaluate $R(s)$ at two different values of $s$ simultaneously via strong additivity of $r$. So consider the vector
\[
a=\left(-\frac12,\frac34,\frac34\right)\in\mathcal{Q}_3\, ,
\]
which has a total negative mass of $s=1/2$. If we use $a$ as an outer vector, insert $a$ into its first block, and use $(1)$ in the other two blocks, the resulting vector is
\[
\Big(-\frac{1}{2}\Big)\left(-\frac12,\frac34,\frac34\right) \oplus \frac{3}{4}(1) \oplus \frac{3}{4}(1) = \left(\frac14,-\frac38,-\frac38,\frac34,\frac34\right)\, .
\]
The total negative mass of this new vector is exactly $3/8 + 3/8 = 3/4$. We then have
\[
\frac{3K}{4} =R\left(\frac34\right)=r(a)-\frac12r(a)+\frac34 r(1)+\frac34 r(1)=\frac12 r(a)=\frac12 R\left(\frac12\right)=\frac{K}{4}\, ,
\]
where the first and final equalities follow from the formula $R(s) = Ks$, and the second equality follows from strong additivity of $r$. It then follows that $K=0$, thus $r(p)=0$ for every $p\in\mathcal{Q}$.

Finally, let $x=(x_1,\ldots,x_n)\in\A$, let $\alpha\in(0,1)$ be such that $|\alpha x_i|\leq1$ for every $i$, and set
\[
y=(\alpha x_1,\ldots,\alpha x_n,1-\alpha)\in\mathcal{Q}_{n+1}\, .
\]
Since $r$ vanishes on $\Q$, we have $r(y)=0$. Moreover, $(\alpha,1-\alpha)\in\Delta_2$, so strong additivity gives
\[
0=r(y)=r(\alpha,1-\alpha)+\alpha r(x)+(1-\alpha)r(1)=\alpha r(x)\, ,
\]
hence $r(x)=0$. It then follows that $h=h^0$, as desired.
\eprf

\bprf[Proof of Theorem~\ref{FPXCXH747}]
Suppose $S:\mathbb{H}^1 \to \mathbb{R}$ satisfies conditions~\ref{FPXCXH7471}-\ref{FPXCXH7473} in the statement of Theorem~\ref{FPXCXH747}. It follows from unitary invariance (condition~\ref{FPXCXH7472}) that in order to show that Eq.~\eqref{FPX91} holds for all $X\in \mathbb{H}^1$, it suffices to show \eqref{FPX91} holds for all diagonal matrices, i.e., matrices of the form $\text{diag}(x_1,\ldots,x_n)$ with $(x_1,\ldots,x_n)\in \A$. Therefore, we now show that for all $X=\text{diag}(x_1,\ldots,x_n)$ with $(x_1,\ldots,x_n)\in \A$ we have
\be \label{RMX17}
S(X)=-\tr(X\log|X|)=-\sum_{i=1}^{n}x_i\log|x_i|\, .
\ee
For this, let $h:\A\to \R$ be the function given by $h(x_1,\ldots,x_n)=S(\operatorname{diag}(x_1,\ldots,x_n))$. By reduction to von~Neumann entropy (condition \ref{FPXCXH7471}), it follows that $h(p_1,\ldots,p_n)=-\sum_{i=1}^{n}p_i\log(p_i)$ for every finite probability vector $(p_1,\ldots,p_n)\in \Delta$. If we can show that $h$ is strongly additive, then it follows from item~\ref{LX5} of Lemma~\ref{LMXA57} that Eq.~\eqref{RMX17} holds for all diagonal matrices $X\in \mathbb{H}^1$.

To show $h$ is strongly additive, let $(p_1,\ldots,p_n)\in \mathcal{Q}_n$, let $q^{(i)}=(q_{i1},\ldots,q_{im_i})\in \A_{m_i}$ for all $i\in \{1,\ldots,n\}$, and let $X_i=\operatorname{diag}(q_{i1},\ldots,q_{im_i})$ for all $i\in \{1,\ldots,n\}$. Since $\sum_j q_{ij} = 1$, each $X_i$ is a unit trace Hermitian matrix, so $X_i\in\mathbb{H}_{m_i}^1$ for all $i\in \{1,\ldots,n\}$. Moreover, since
\be \label{AFXCMB67}
\operatorname{diag}\Big(\bigoplus_{i=1}^{n}p_iq^{(i)}\Big)=\bigoplus_{i=1}^n p_iX_i\, ,
\ee
we have
\begin{align*}
h\Big(\bigoplus_{i=1}^{n}p_iq^{(i)}\Big)&=S\left(\text{diag}\Big(\bigoplus_{i=1}^{n}p_iq^{(i)}\Big)\right)=S\Big(\bigoplus_{i=1}^n p_iX_i\Big) \\
&=S(\text{diag}(p_1,\ldots,p_n))+\sum_{i=1}^{n}p_iS(X_i) \\
&=S(\text{diag}(p_1,\ldots,p_n))+\sum_{i=1}^{n}p_iS(\text{diag}(q_{i1},\ldots,q_{im_i})) \\
&=h(p_1,\ldots,p_n)+\sum_{i=1}^{n}p_ih\big(q^{(i)}\big)\, ,
\end{align*}
where the third equality follows from the strong additivity of $S$ (condition~\ref{FPXCXH7473}). It then follows that $h$ is strongly additive, thus Eq~\eqref{RMX17} indeed holds for all diagonal matrices $X\in \mathbb{H}^1$, which by unitary invariance implies Eq.~\eqref{FPX91} holds for all $X\in \mathbb{H}^1$. 

Conversely, let $S:\mathbb{H}^1\to \R$ be the function given by $S(X)=-\tr(X\log|X|)$. It is immediate from its definition that $S$ satisfies reduction to von~Neumann entropy (condition~\ref{FPXCXH7471}). To show unitary invariance (condition~\ref{FPXCXH7472}), let $X\in\mathbb{H}^1$ and let $X=V\operatorname{diag}(\lambda_1,\ldots,\lambda_n)V^\dagger$ be a unitary diagonalization of $X$. We then have
\[
X\log|X|=V\operatorname{diag}\big(\lambda_1\log|\lambda_1|,\ldots,\lambda_n\log|\lambda_n|\big)V^\dagger\, ,
\]
which implies $S(X)$ may be given in terms of the eigenvalues of $X$ as
\be \label{EIGENFXN87}
S(X)=-\sum_{i=1}^n\lambda_i\log|\lambda_i|\, .
\ee
It then follows that $S(UXU^\dagger)=S(X)$ for every unitary $U$ which is the same size as $X$, thus $S$ satisfies unitary invariance (condition~\ref{FPXCXH7472}). Finally, since $S$ satisfies the eigenvalue formula \eqref{EIGENFXN87}, the proof that $S$ satisfies strong additivity (condition~\ref{FPXCXH7473}) is essentially the same as the proof of item~\ref{LX2} of Lemma~\ref{LMXA57}, thus completing the proof.
\eprf

In light of Theorem~\ref{FPXCXH747}, we will refer to the entropy functional $S$ as given by \eqref{FPX91} as the \emph{PDM entropy}. While the PDM entropy has appeared in the literature in various contexts~\cite{PaTK23,SSW14,TTC22,CiKu23,FuPa23} (as well as its extension to normal operators by defining $|X|=\sqrt{XX^{\dag}}$), the quantum information-theoretic meaning of the PDM entropy is not yet well-understood. Nevertheless, Theorem~\ref{FPXCXH747} yields precise mathematical justification for the use of PDM entropy as an extension of von~Neumann entropy to pseudo-density matrices. In particular, as von~Neumann entropy satisfies strong additivity with respect to convex combinations, it is only natural that any extension of von~Neumann entropy to pseudo-density matrices should satisfy strong additivity with respect to affine combinations within the interval $[-1,1]$ where its eigenvalues reside.   

\section{Subadditivity of PDM entropy}

Apart from unitary invariance and strong additivity, it was established in Ref.~\cite{FuPa23} that PDM entropy shares other properties in common with von~Neumann entropy, such as additivity over products, namely, 
\[
S(\tau\otimes \upsilon)=S(\tau)+S(\upsilon)\, ,
\]
and a Fannes-Audenaert type inequality~\cite{Fannes_1973,Audenaert_2007}. While subadditivity 
\[
S(\rho_{AB})\leq S(\rho_A)+S(\rho_B)
\]
is a fundamental property of von~Neumann entropy for bipartite density matrices $\rho_{AB}$, it is known that there exists unit trace, bipartite Hermitian matrices $\tau_{AB}$ which violate subadditivity. However, it was conjectured in Ref.~\cite{FuPa23} that if $\varrho_{AB}$ is in fact a 2-time pseudo-density matrix as defined by Eq.~\eqref{GBLXFXM627}, then subadditivity holds. We now prove that this conjecture holds for 2-time pseudo-density matrices whose initial state is that of a single qubit. In particular, we will prove the following:

\begin{theorem}
\label{thm:qubit-input-sot-mutual-information}
Let $\E:M_2(\mathbb C)\to M_d(\mathbb C)$ be a completely positive, trace-preserving map from $2\times 2$ complex matrices to $d\times d$ complex matrices for some $d>0$, let $\rho_A\in M_2(\mathbb C)$ be a density matrix, and let $\varrho_{AB}$ be a 2-time pseudo-density matrix given by $\varrho_{AB}=\frac{1}{2}\{\rho\otimes \mathds{1}_d, \J[\E] \}$. Then
\[
S(\varrho_{AB})\leq S(\rho_A)+S(\rho_B)\, ,
\]
where $\rho_B\equiv \tr_A(\varrho_{AB})=\E(\rho_A)$.
\end{theorem}

We first prove some preliminary results.

\begin{lemma}
\label{lem:log-absolute-resolvent-representation}
Let $X$ be a nonsingular Hermitian matrix. Then
\[
\log|X| = \frac12 \int_0^\infty \left( \frac{1}{1+s}\mathds{1}-(X^2+s\mathds{1})^{-1} \right)ds\, .
\]
\end{lemma}

\begin{proof}
For $y>0$,
\[
\log y
=
\int_0^\infty
\left(
\frac{1}{1+s}-\frac{1}{y+s}
\right)ds\, .
\]
Indeed, integrating up to $R>0$ gives
\[
\int_0^R
\left(
\frac{1}{1+s}-\frac{1}{y+s}
\right)ds
=
\log(1+R)-\log(y+R)+\log y\, ,
\]
which converges to $\log y$ as $R\to\infty$. Applying this identity to
$y=x^2$ gives
\[
\log|x|
=
\frac12 \int_0^\infty \left( \frac{1}{1+s}-\frac{1}{x^2+s} \right)ds,
\qquad x\ne0\, .
\]
Now let $X=\sum_k\lambda_kP_k$ be the spectral decomposition of $X$. Since $X$ is nonsingular, each $\lambda_k$ is nonzero. Moreover,
\[
X^2+s\mathds{1}=\sum_k(\lambda_k^2+s)P_k\, ,
\]
and hence
\[
(X^2+s\mathds{1})^{-1}
=
\sum_k\frac{1}{\lambda_k^2+s}P_k \, .
\]
Therefore,
\begin{align*}
\log|X|&=\sum_k \log|\lambda_k|P_k=\frac12\int_0^\infty
\sum_k
\left(
\frac{1}{1+s}
-
\frac{1}{\lambda_k^2+s}
\right)P_k\,ds\\
&=\frac12\int_0^\infty
\left(
\sum_k\frac{P_k}{1+s}
-
\sum_k\frac{P_k}{\lambda_k^2+s}
\right)ds\\
&=\frac12\int_0^\infty
\left(
\frac{1}{1+s}\mathds{1}
-
(X^2+s\mathds{1})^{-1}
\right)ds \, ,
\end{align*}
as desired.
\end{proof}

\begin{lemma}
\label{lem:reduced-trace-inequality}
Let $R,K,C$ be operators on a finite-dimensional Hilbert space such that $R,K\ge0$ and $C=C^\dagger$,  and suppose that $u\geq 0$ is such that
$\mathds{1}-uR\ge0$ and $C+uK\ge0$, and let $L=(\mathds{1}+C^2)^{-1}$. Then
\[
\tr[KL]+\mathfrak{Re}\,\tr[CLR]\ge0\, .
\]
\end{lemma}

\begin{proof}
If $u=0$, then $C\ge0$. Since $L$ is a positive function of $C$, we have $L\ge0$, $CL=LC$, and $CL\ge0$. Hence $\tr[KL]\ge0$ and $\tr[CLR]\ge0$. Now assume $u>0$, diagonalize $C$, and write
\[
C=\operatorname{diag}(c_1,\ldots,c_n)\, ,
\qquad
L=\operatorname{diag}(\ell_1,\ldots,\ell_n)\, ,
\qquad
\ell_j=\frac{1}{1+c_j^2}>0\, ,
\]
so that
\[
\tr[KL]+\mathfrak{Re}\,\tr[CLR]=\sum_j\ell_jK_{jj}+\sum_jc_j\ell_jR_{jj}=\sum_j\ell_j(K_{jj}+c_jR_{jj})\, .
\]
Since $\ell_j>0$ for all $j$, it suffices to prove that $K_{jj}+c_jR_{jj}\geq 0$ for all $j$. Now since $K,R\ge0$, we have $K_{jj},R_{jj}\ge0$. Moreover, $\mathds{1}-uR\ge0$ gives
\[
R_{jj}\le\frac1u\, ,
\]
while $C+uK\ge0$ gives
\[
K_{jj}\ge-\frac{c_j}{u}\, .
\]
Therefore, if $c_j\ge0$, then $K_{jj}+c_jR_{jj}\ge0$. If $c_j<0$, then
\[
K_{jj}+c_jR_{jj}
\ge
-\frac{c_j}{u}+\frac{c_j}{u}=0\, ,
\]
thus for all $j$ we have $K_{jj}+c_jR_{jj}\geq 0$, as desired.
\end{proof}

\begin{lemma} \label{prop:two-block-signed-pinching}
Let $P,Q\ge0$ be positive operators on the same finite-dimensional Hilbert space, let $D$ be an arbitrary operator on that space, and for $0\le t\le1$ set
\[
X_t=
\begin{pmatrix}
P&tD\\
tD^\dagger&Q
\end{pmatrix}\, ,
\qquad \text{and} \qquad
Z=
\begin{pmatrix}
0&D\\
D^\dagger&0
\end{pmatrix}\, .
\]
Then the following statements hold.
\begin{enumerate}[i.]
\item \label{LSX0}
The function $t \mapsto S(X_t)$ is differentiable at every $t_0 \in [0,1]$ for which $X_{t_0}$ is nonsingular,
where $S$ is the PDM entropy.
\item \label{LSX1}
At every $t_0$ for which $X_{t_0}$ is nonsingular,
\begin{equation} \label{eq:entropy-derivative-resolvent}
\left.\frac{d}{dt}S(X_t)\right|_{t=t_0}=\frac12 \int_0^\infty \tr\left[Z(X_{t_0}^2+s\mathds{1})^{-1}\right]ds\, ,
\end{equation}
where $S$ is the PDM entropy.
\item \label{LSX2}
For every $a>0$ and every $0\le t\le1$,
\begin{equation} \label{eq:resolvent-sign-main}
\tr\left[Z(X_t^2+a^2\mathds{1})^{-1}\right]\le0\, .
\end{equation}
\item \label{LSX3}
$S
\begin{pmatrix}
P&D\\
D^\dagger&Q
\end{pmatrix}
\le
S(P\oplus Q)
$, where $S$ is the PDM entropy.
\end{enumerate}
\end{lemma}
\begin{proof}
The family $X_t$ gives a continuous interpolation between the block-diagonal operator $X_0=P\oplus Q$ and the full block operator 
\[
X_1=
\begin{pmatrix}P&D\\ 
D^\dagger&Q
\end{pmatrix}\, ,
\]
thus item~\ref{LSX3} is the inequality $S(X_1)\le S(X_0)$. We first establish differentiability and the trace derivative formula in item~\ref{LSX0}, the resolvent integral formula in item~\ref{LSX1}, and the resolvent estimate in item~\ref{LSX2}, and then use them to prove item~\ref{LSX3}.

\underline{Item~\ref{LSX0}}: Let $t_0 \in [0,1]$ be such that $X_{t_0}$ is nonsingular, and let $\sigma(X_t)$ denote the spectrum of $X_t$ for all $t\in [0,1]$. Because the matrix entries of $X_t = X_0 + tZ$ are affine in $t$, the roots of its characteristic polynomial---and thus the eigenvalues of $X_t$---depend continuously on $t$. Since $0 \notin \sigma(X_{t_0})$, there exists a neighborhood $(t_0-\delta, t_0+\delta)$ such that for all $t$ in this interval, the spectrum $\sigma(X_t)$ is strictly bounded away from zero. Consequently, there exists a compact set $K \subset \mathbb{R} \setminus \{0\}$ containing $\sigma(X_t)$ for all $t \in (t_0-\delta, t_0+\delta)$.

The function $f(x) = -x\log|x|$ is continuously differentiable on $K$. By the Weierstrass approximation theorem, we can choose a sequence of polynomials $P_n(x)$ converging uniformly to $f(x)$ on $K$, such that their derivatives $P_n'(x)$ converge uniformly to $f'(x)$ on $K$. For each $n$, the scalar function $g_n(t) = \tr[P_n(X_t)]$ is a linear combination of terms of the form $\tr(X_t^k)$. By the product rule and the cyclic property of the trace,
\[
\frac{d}{dt}\tr(X_t^k) = \tr\left(\sum_{j=0}^{k-1} X_t^j \frac{dX_t}{dt} X_t^{k-1-j}\right) = \tr\left(\sum_{j=0}^{k-1} X_t^j Z X_t^{k-1-j}\right) = k\tr(X_t^{k-1} Z)\, .
\]
By linearity, it follows that $g_n(t)$ is differentiable with derivative $g_n'(t) = \tr[P_n'(X_t)Z]$. Due to the uniform convergence on $K$, $g_n(t)$ converges uniformly to $S(X_t) = \tr[f(X_t)]$, and $g_n'(t)$ converges uniformly to $\tr[f'(X_t)Z]$. It then follows from standard analysis that the limit function $S(X_t)$ is differentiable at $t_0$, as desired. 

\underline{Item~\ref{LSX1}}: As established in the proof of item~\ref{LSX0}, the derivative of $S(X_t)$ at $t_0$ is given by the limit of the polynomial derivatives, which yields
\[
\left.\frac{d}{dt}S(X_t)\right|_{t=t_0} = \tr[f'(X_{t_0})Z]\, .
\]
Substituting $f'(x) = -(\log|x| + 1)$ gives
\[
\left.\frac{d}{dt}S(X_t)\right|_{t=t_0} = -\tr\left[Z(\log|X_{t_0}| + \mathds{1})\right] = -\tr[Z\log|X_{t_0}|]\, ,
\]
where we utilized $\tr[Z]=0$. Applying Lemma~\ref{lem:log-absolute-resolvent-representation} and again using $\tr[Z]=0$, we obtain
\begin{align*}
\left.\frac{d}{dt}S(X_t)\right|_{t=t_0}&=-\frac12\tr\left[Z\int_0^\infty\left(\frac{1}{1+s}\mathds{1}-(X_{t_0}^2+s\mathds{1})^{-1}\right)ds\right]\\
&=\frac12\int_0^\infty\tr\left[Z(X_{t_0}^2+s\mathds{1})^{-1}\right]ds\, ,
\end{align*}
which proves item~\ref{LSX1}.

\underline{Item~\ref{LSX2}}: Fix $a>0$ and define $M_t=a\mathds{1}+iX_t$ for $0\le t\le1$. Since $X_t$ is Hermitian and $a>0$, $M_t$ is invertible, with
\[
M_t^{-1}=(a\mathds{1}-iX_t)(a^2\mathds{1}+X_t^2)^{-1}\, .
\]
Since $X_t(a^2\mathds{1}+X_t^2)^{-1}$ and $Z$ are Hermitian, $\tr\left(X_t(a^2\mathds{1}+X_t^2)^{-1}Z\right)$ is real. Hence
\begin{align*}
\mathfrak{Im}\,\tr(M_t^{-1}iZ)
&=\mathfrak{Im}\,\tr\left((X_t+ia\mathds{1})(a^2\mathds{1}+X_t^2)^{-1}Z\right)\notag\\
&=a\tr\left(Z(X_t^2+a^2\mathds{1})^{-1}\right)\, .
\end{align*}
Since $a>0$, \eqref{eq:resolvent-sign-main} is equivalent to
\begin{equation} \label{RSVSGN55}
\mathfrak{Im}\,\tr(M_t^{-1}iZ)\le0\, .
\end{equation}

To prove \eqref{RSVSGN55}, let $A_a=a\mathds{1}+iP$, so that
\[
M_t=\begin{pmatrix}A_a&itD\\ itD^\dagger&a\mathds{1}+iQ\end{pmatrix}\, .
\]
For a block matrix $M=\begin{pmatrix}A&B\\ C&E\end{pmatrix}$ with $A$ invertible, we denote the Schur complement of $A$ in $M$ by $M/A=E-CA^{-1}B$, for which $\det(M)=\det(A)\det(M/A)$. Since $A_a$ is invertible and
\[
A_a^{-1}=(a\mathds{1}-iP)(a^2\mathds{1}+P^2)^{-1}\, ,
\]
the Schur complement of $A_a$ in $M_t$ is
\[
M_t/A_a=a\mathds{1}+iQ+t^2D^\dagger A_a^{-1}D=a\mathds{1}+t^2R+i(Q-t^2K)=G_{t^2}\, ,
\]
where
\[
R=D^\dagger a(a^2\mathds{1}+P^2)^{-1}D\ge0,\qquad \text{and} \qquad K=D^\dagger P(a^2\mathds{1}+P^2)^{-1}D\ge0\, .
\]
Consequently,
\[
\det(M_t)=\det(A_a)\det(G_{t^2})\, .
\]

For a differentiable family of invertible matrices $N_t$, Jacobi's determinant formula is
\begin{equation} \label{JACDET}
\frac{d}{dt}\det(N_t)=\det(N_t)\tr\left(N_t^{-1}\frac{dN_t}{dt}\right)\, .
\end{equation}
Since $M_t$ and $A_a$ are invertible, the identity $\det(M_t)=\det(A_a)\det(G_{t^2})$ shows that $G_{t^2}$ is also invertible. Applying \eqref{JACDET} to $M_t$ and $G_{t^2}$ gives
\[
\tr(M_t^{-1}iZ)=\frac{\frac{d}{dt}\det(M_t)}{\det(M_t)}=\frac{\frac{d}{dt}\det(G_{t^2})}{\det(G_{t^2})}=\tr\left(G_{t^2}^{-1}\frac{d}{dt}G_{t^2}\right)\, .
\]
Since $\frac{d}{dt}G_{t^2}=2t(R-iK)$, we obtain
\begin{equation} \label{JACSHR}
\mathfrak{Im}\,\tr(M_t^{-1}iZ)=2t\,\mathfrak{Im}\,\tr\left(G_{t^2}^{-1}(R-iK)\right)\, .
\end{equation}
For $t=0$, \eqref{RSVSGN55} follows immediately from \eqref{JACSHR}. For $t>0$, \eqref{RSVSGN55} is equivalent to
\[
\mathfrak{Im}\,\tr\left(G_{t^2}^{-1}(R-iK)\right)\le0\, .
\]
For $u>0$, define
\[
G_u=a\mathds{1}+uR+i(Q-uK)\, .
\]
Then $G_u$ is invertible. Indeed, if $G_uv=0$, then
\[
0=\mathfrak{Re}\,\langle v,G_uv\rangle=a\langle v,v\rangle+u\langle v,Rv\rangle\, .
\]
Since $a>0$ and $R\ge0$, this implies $v=0$. It therefore suffices to prove
\begin{equation} \label{GUSGN}
\mathfrak{Im}\,\tr\left(G_u^{-1}(R-iK)\right)\le0 \qquad \forall u>0\, .
\end{equation}
For this, set $H_u=a\mathds{1}+uR>0$ and normalize $R$, $K$, and $Q$ with respect to $H_u$ to obtain
\[
\widetilde R=H_u^{-1/2}RH_u^{-1/2},\qquad \widetilde K=H_u^{-1/2}KH_u^{-1/2},\qquad \text{and} \qquad \widetilde Q=H_u^{-1/2}QH_u^{-1/2}\, .
\]
These operators are positive, and
\[
\mathds{1}-u\widetilde R=H_u^{-1/2}(H_u-uR)H_u^{-1/2}=aH_u^{-1}\ge0\, .
\]
Define $C_u=\widetilde Q-u\widetilde K$. Then $C_u=C_u^\dagger$, $C_u+u\widetilde K=\widetilde Q\ge0$, and
\[
G_u=H_u^{1/2}(\mathds{1}+iC_u)H_u^{1/2}\, .
\]
If $L_u=(\mathds{1}+C_u^2)^{-1}$, then $(\mathds{1}+iC_u)^{-1}=(\mathds{1}-iC_u)L_u$. Since $L_u$ is a function of $C_u$, the operators $L_u$ and $C_uL_u$ are Hermitian. Hence $\tr(L_u\widetilde R)$ and $\tr(C_uL_u\widetilde K)$ are real. Using cyclicity of the trace, we therefore obtain
\begin{align*}
\mathfrak{Im}\,\tr\left(G_u^{-1}(R-iK)\right)
&=\mathfrak{Im}\,\tr\left((\mathds{1}+iC_u)^{-1}(\widetilde R-i\widetilde K)\right)\\
&=\mathfrak{Im}\,\tr\left(L_u\widetilde R-iL_u\widetilde K-iC_uL_u\widetilde R-C_uL_u\widetilde K\right)\\
&=-\tr(\widetilde K L_u)-\mathfrak{Re}\,\tr(C_uL_u\widetilde R)\, .
\end{align*}
The operators $\widetilde R$, $\widetilde K$, and $C_u$ satisfy the hypotheses of Lemma~\ref{lem:reduced-trace-inequality}, and therefore
\[
\tr(\widetilde K L_u)+\mathfrak{Re}\,\tr(C_uL_u\widetilde R)\ge0\, .
\]
It then follows that for all $u>0$ we have
\[
\mathfrak{Im}\,\tr\left(G_u^{-1}(R-iK)\right)=-\tr(\widetilde K L_u)-\mathfrak{Re}\,\tr(C_uL_u\widetilde R)\le0\, ,
\]
thus \eqref{GUSGN} holds, as desired.

\underline{Item~\ref{LSX3}}: Let $\varepsilon>0$ and define the regularized matrix
\[
X_t^\varepsilon = X_t + \varepsilon\mathds{1} = \begin{pmatrix}P+\varepsilon\mathds{1}&tD\\ tD^\dagger&Q+\varepsilon\mathds{1}\end{pmatrix} \, .
\]
Notice that $X_t^\varepsilon$ takes the exact same block form as $X_t$, with positive operators $P_\varepsilon = P+\varepsilon\mathds{1}$ and $Q_\varepsilon = Q+\varepsilon\mathds{1}$. At $t=0$, we have $X_0^\varepsilon = P_\varepsilon \oplus Q_\varepsilon$. Because $P, Q \ge 0$, the blocks $P_\varepsilon$ and $Q_\varepsilon$ are strictly positive definite, meaning $X_0^\varepsilon$ is nonsingular and $\det(X_0^\varepsilon) > 0$.

The determinant $\det(X_t^\varepsilon)$ is a polynomial in $t$. Because it evaluates to a strictly positive number at $t=0$, it is not the zero polynomial, and therefore has at most finitely many roots in the interval $[0,1]$. Let these roots be $s_1 < s_2 < \dots < s_k$. These finitely many points partition $[0,1]$ into a finite number of open subintervals. 

On each such open subinterval, $X_t^\varepsilon$ is nonsingular. Thus, by items~\ref{LSX1} and \ref{LSX2}, we have
\[
\frac{d}{dt}S(X_t^\varepsilon)\le0\, .
\]
This implies that $S(X_t^\varepsilon)$ is non-increasing on each open subinterval. Furthermore, because the map $X \mapsto S(X)$ is continuous on the space of finite-dimensional Hermitian matrices, the composition $t \mapsto S(X_t^\varepsilon)$ is continuous on the entire closed interval $[0,1]$. A continuous function that is non-increasing on the interior of adjacent subintervals is non-increasing everywhere, which yields
\[
S(X_1^\varepsilon)\le S(X_0^\varepsilon) \, .
\]
Finally, since $X_j^\varepsilon \to X_j$ as $\varepsilon\to 0$ for $j=0,1$, the continuity of $S$ on the space of Hermitian matrices guarantees that the inequality is preserved in the limit, which yields
\[
S(X_1)\le S(X_0) \, ,
\]
as desired.
\end{proof}

\begin{proof}[Proof of Theorem~\ref{thm:qubit-input-sot-mutual-information}]
Choose an orthonormal eigenbasis $\{\ket{0},\ket{1}\}$ of $\rho_A$ and write
\[
\rho_A=p\dyad{0}{0}+q\dyad{1}{1}\, ,\qquad p,q\ge0,\qquad p+q=1 \, .
\]
We also set
\[
U=\E(\dyad{0}{0})\, ,\qquad V=\E(\dyad{1}{1})\, ,\qquad \text{and} \qquad W=\E(\dyad{1}{0}) \, .
\]
Since $\E$ is completely positive and trace preserving, $U$ and $V$ are valid density matrices. By evaluating the anti-commutator $\varrho_{AB}=\frac{1}{2}\{\rho_A\otimes \mathds{1}_d, \J[\E]\}$ in this basis we then have
\[
\varrho_{AB}=
\begin{pmatrix}
pU&\frac{p+q}{2}W\\ 
\frac{p+q}{2}W^\dagger&qV
\end{pmatrix}
=
\begin{pmatrix}
pU&\frac12 W\\ 
\frac12 W^\dagger&qV
\end{pmatrix}\, ,
\]
where the off-diagonal simplifies as $p+q=1$. Applying Lemma~\ref{prop:two-block-signed-pinching} (\ref{LSX3}) with $P=pU$, $Q=qV$, and $D=W/2$ then yields
\[
S(\varrho_{AB})\le S(pU\oplus qV) \, .
\]
Moreover, since
\begin{align*}
S(pU\oplus qV)&=S(\text{diag}(p,q))+pS(U)+qS(V) && \text{by strong additivity} \\
&=S(\rho_A)+pS(U)+qS(V) && \text{since $\rho_A=p\dyad{0}{0}+q\dyad{1}{1}$} \\
&\leq S(\rho_A)+S(pU+qV) && \text{by concavity of von~Neumann entropy} \\
&=S(\rho_A)+S(\rho_B) && \text{since $\rho_B=pU+qV$}
\end{align*}
we then have
\[
S(\varrho_{AB})\leq S(pU\oplus qV)\leq S(\rho_A)+S(\rho_B) \, ,
\]
as desired.
\end{proof}

\section{Some fundamental examples}

In this section we consider some fundamental examples of dynamics in the context of quantum information, and obtain some results regarding the entropy of the associated pseudo-density matrices. Given a state-channel pair $(\rho,\E)$, we will denote the associated 2-time pseudo-density matrix by $\E\star \rho$, so that
\[
\E\star \rho=\frac{1}{2}\big\{\rho\otimes \mathds{1}\, ,\J[\E]\big\}\, .
\] 
In such a case, we will often tacitly assume that the input and output systems of the channel $\E$ are labeled as $A$ and $B$, respectively, even if not explicitly stated as such. The algebra of complex $d\times d$ matrices will be denoted by $M_d(\C)$.

We first consider the case of unitary evolution. The following result was first proved in Ref.~\cite{FuPa23}, which may be interpreted as a dynamical formulation of the fact that unitary evolution does not introduce any further uncertainty into the state of a quantum system. In fact, this result was one of our original motivations for the further investigation of what we now refer to as PDM entropy. 

\bn[\cite{FuPa23}] \label{UNTRXY71}
Suppose $\mathcal{U}:M_d(\C)\to M_d(\C)$ is a unitary channel, so that there exists a unitary operator $U$ on $\C^d$ such that $\mathcal{U}(X)=UXU^{\dag}$ for every $X\in M_d(\C)$. Then $S(\mathcal{U}\star \rho)=S(\rho)$ for every density matrix $\rho\in M_d(\C)$.
\en

We now consider the case of depolarizing channels, which model quantum systems undergoing isotropic decoherence due to system-environment interactions.

\bn
\label{DPLQBT58}
Consider the qubit depolarizing channel
\[
\mathcal D_p(X)=pX+(1-p)\tr(X)\mathds{1}_2/2,\qquad 0\le p\le1\, ,
\]
so that $\mathcal D_0$ is the completely depolarizing channel and $\mathcal D_1$ is the identity channel. For every density matrix $\rho\in M_2(\mathbb C)$, the function $p\mapsto S(\mathcal D_p\star\rho)$ is non-increasing on $[0,1]$, and
\[
S(\rho)\le S(\mathcal D_p\star\rho)\le S(\rho)+\log2\, ,
\]
with the upper and lower bounds being attained for $p=0$ and $p=1$, respectively. 
\en

\begin{proof}
For a composition of channels $\F\circ \E$ the associated Jamio{\l}kowski operator satisfies
\[
\J[\mathcal F\circ\E]=(\id_A\otimes\mathcal F)(\J[\E])\, ,
\]
thus
\begin{equation} \label{STCMP67}
\begin{aligned}
(\F\circ\E)\star\rho&=\frac{1}{2}\big\{\rho\otimes \mathds{1}\, ,\J[\F\circ \E]\big\}=\frac{1}{2}\big\{\rho\otimes\mathds{1},(\id_A\otimes\mathcal F)(\J[\E])\big\} \\
&=(\id_A\otimes\mathcal F)\left(\frac{1}{2}\big\{\rho\otimes\mathds{1},\J[\E]\big\}\right)=(\id_A\otimes\mathcal F)(\E\star\rho)\, .
\end{aligned}
\end{equation}
Moreover, if $t>0$ and $X$ is Hermitian with eigenvalues $\lambda_i$, then
\begin{equation} \label{ENTSCL54}
S(tX)=-\sum_i t\lambda_i\log|t\lambda_i|=tS(X)-t\log t\sum_i\lambda_i=tS(X)-t\log t\,\tr(X)\, .
\end{equation}

Now choose an eigenbasis of $\rho$, write $\rho=a\dyad{0}{0}+(1-a)\dyad{1}{1}$ where $0\le a\le1$, and set $\theta=(1+p)/2$. In this basis, we have
\[
\mathcal D_p\star\rho=
\begin{pmatrix}
a\theta&0&0&0\\
0&a(1-\theta)&p/2&0\\
0&p/2&(1-a)(1-\theta)&0\\
0&0&0&(1-a)\theta
\end{pmatrix}.
\]
Set $s=\sqrt{a(1-a)}$. Since the two outer scalar blocks and the diagonal entries of the middle $2\times 2$ block of $\mathcal D_p\star\rho$ are nonnegative, positivity is equivalent to non-negativity of the determinant of the middle $2\times 2$ block. This determinant is $(s^2(1-p)^2-p^2)/4$, so $\mathcal D_p\star\rho\ge0$ precisely when $p\le p_0=s/(1+s)$. If $0\le p\le q\le p_0$ with $q>0$, then $\mathcal D_p=\mathcal D_{p/q}\circ\mathcal D_q$, and~\eqref{STCMP67} gives
\[
\mathcal D_p\star\rho=(\id_A\otimes\mathcal D_{p/q})(\mathcal D_q\star\rho)\, .
\]
Since $q\le p_0$, the preceding positivity criterion gives $\mathcal D_q\star\rho\ge0$, while the displayed matrix gives $\tr(\mathcal D_q\star\rho)=1$. Thus $\mathcal D_q\star\rho$ is a density matrix. Moreover, $\id_A\otimes\mathcal D_{p/q}$ is a unital quantum channel. As a unital quantum channel maps a density matrix to a state majorized by the input, and the von~Neumann entropy is Schur-concave, it follows that
\[
S(\mathcal D_p\star\rho)\ge S(\mathcal D_q\star\rho)\, ,
\]
thus $p\mapsto S(\mathcal D_p\star\rho)$ is non-increasing on $[0,p_0]$. For $p_0<p<1$, set
\[
\Xi=\sqrt{(1-p)^2(1-4s^2)+4p^2}\, ,\qquad u=\frac{1-p+\Xi}{4}\, ,\qquad \text{and} \qquad v=\frac{\Xi-1+p}{4}\, .
\]
Then $\sigma(\mathcal D_p\star\rho)=\{a\theta,(1-a)\theta,u,-v\}$ with $0<v\le u$, and
\[
S(\mathcal D_p\star\rho)=\theta S(\rho)-\theta\log\theta-u\log u+v\log v\, .
\]
Since $u-v=(1-p)/2$ and $uv=(p^2-s^2(1-p)^2)/4$,
\[
\frac{d}{dp}S(\mathcal D_p\star\rho)=\frac12\left(S(\rho)+\log\frac{u}{\theta}\right)+v'\log\frac{v}{u}\, ,
\]
where $v'=\frac{p+s^2(1-p)+v}{2(u+v)}\ge0$. For $0<a<1$, concavity of the logarithm gives
\[
S(\rho)=2\left(a\log a^{-1/2}+(1-a)\log(1-a)^{-1/2}\right)\le2\log\left(\sqrt a+\sqrt{1-a}\right)=\log(1+2s)\, ,
\]
with the boundary cases following by continuity. Set $t=p/(1-p)$ and $y=(1+\sqrt{1+4(t^2-s^2)})/2$. Since $p\ge p_0$ is equivalent to $t\ge s$, we have $u/\theta=y/(1+2t)$. Moreover, $s\le1/2$ implies $\sqrt{1+4(t^2-s^2)}\ge2t$, and hence $y-t\ge1/2\ge2s^2$. Therefore
\[
\frac{d}{dt}\frac{y}{1+2t}=\frac{2(t-y+2s^2)}{(1+2t)^2(2y-1)}\le0\, .
\]
Thus $t\mapsto y/(1+2t)$ is non-increasing on $[s,\infty)$, and since $y(s)=1$,
\[
\frac{u}{\theta}\le\frac{1}{1+2s}\, .
\]
It then follows that
\[
S(\rho)+\log\frac{u}{\theta}\le\log(1+2s)-\log(1+2s)=0\, ,
\]
and since $v'\ge0$ and $0<v\le u$ we have $v'\log\frac{v}{u}\leq 0$. Therefore, we have 
\[
\frac{d}{dp}S(\mathcal D_p\star\rho)\leq 0
\]
for $p_0<p<1$, hence the function $p\mapsto S(\mathcal D_p\star\rho)$ is non-increasing on $(p_0,1)$. By continuity this monotonicity extends to $[p_0,1]$. Taken together with the monotonicity already proved on $[0,p_0]$, it follows that $p\mapsto S(\mathcal D_p\star\rho)$ is non-increasing on all of $[0,1]$. Finally, we have
\[
S(\mathcal D_0\star\rho)=S(\rho)+\log2\, ,
\qquad \text{and} \qquad
S(\mathcal D_1\star\rho)=S(\rho)\, ,
\]
which establishes the upper and lower bounds, thus concluding the proof.
\end{proof}

\begin{figure}[htbp]
\centering
\includegraphics[width=0.8\textwidth]{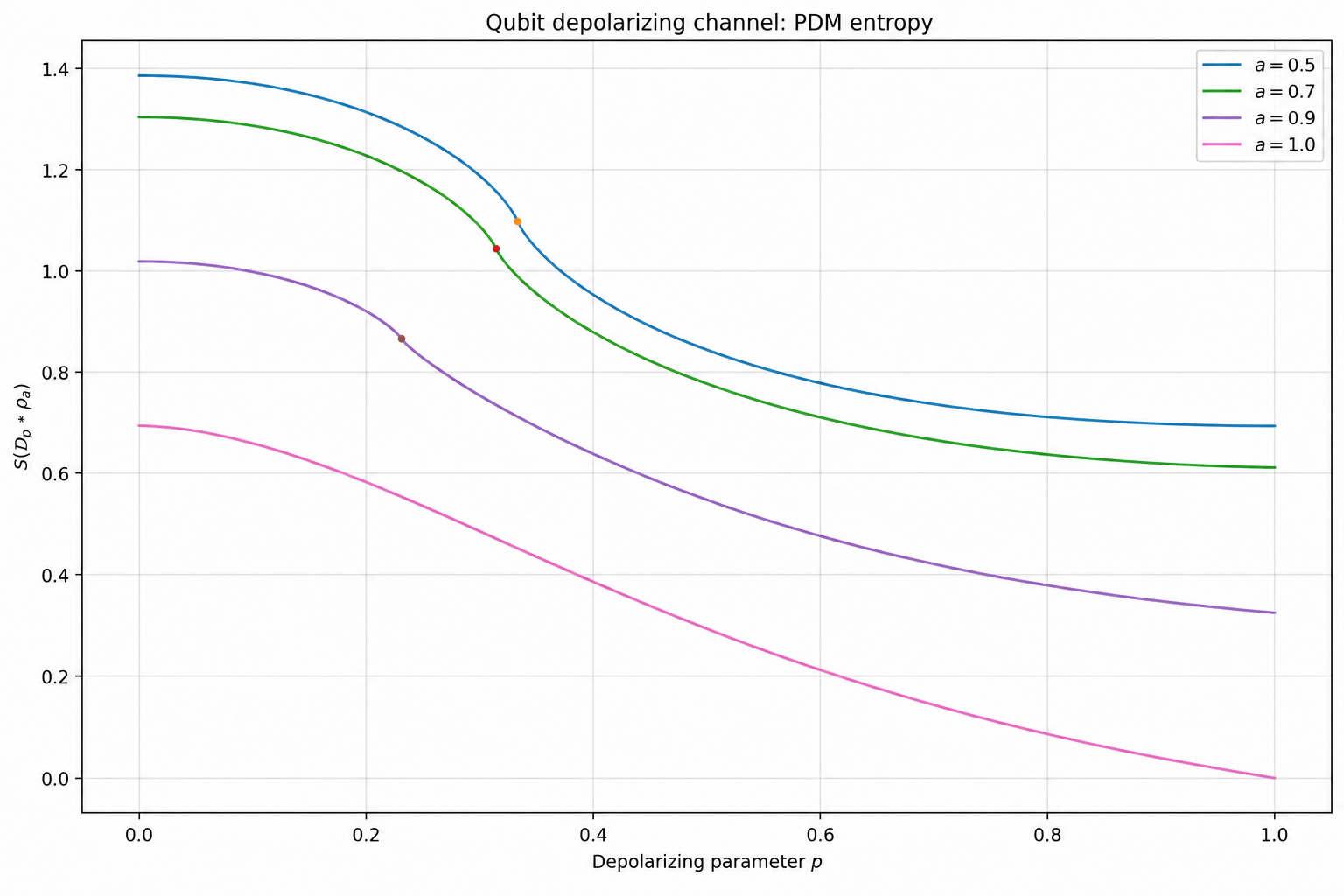}
\caption{The PDM entropy $S(\mathcal D_p\star\rho_a)$ for the qubit depolarizing channel, where $\rho_a=a\dyad{0}{0}+(1-a)\dyad{1}{1}$, for $a=0.5,0.7,0.9,1$. The entropy is nonincreasing in $p$, in agreement with the proposition above. The marked points indicate the values of $p$ at which $\mathcal D_p\star\rho_a$ acquires a zero eigenvalue.}
\label{fig:qubit-depolarizing-entropy}
\end{figure}


In Figure~\ref{fig:qubit-depolarizing-entropy} we plot the PDM entropy for several values of $a$ as a function of the depolarizing parameter $p\in [0,1]$. The following result establishes bounds on the PDM entropy for depolarizing channels on systems of arbitrary dimension.

\bn
Let $\mathcal D_p:M_d(\mathbb C)\to M_d(\mathbb C)$ be the $d$-dimensional depolarizing channel defined by $\mathcal D_p(X)=pX+(1-p)\tr(X)\mathds{1}_d/d$ for $0\le p\le1$. Then, for every density matrix $\rho\in M_d(\mathbb C)$,
\[
S(\rho)\le S(\mathcal D_p\star\rho)\le S(\rho)+\log d\, .
\]
\en

\begin{proof}
Choose an eigenbasis of $\rho$ and write $\rho=\sum_i r_i\dyad{i}{i}$. Set $b=(1-p)/d$ and $c=p+b$. Then $\mathcal D_p\star\rho$ is the direct sum of the scalar blocks $cr_i$ and, for $i<j$, the blocks 
\[
B_{ij}=
\begin{pmatrix}
br_i&\frac{p}{2}(r_i+r_j)\\
\frac{p}{2}(r_i+r_j)&br_j
\end{pmatrix}.
\]
By additivity of $S$ over direct sums,
\[
S(\mathcal D_p\star\rho)=\sum_i f(cr_i)+\sum_{i<j}S(B_{ij})\, ,
\]
where $f(x)=-x\log|x|$. For each $i<j$, Lemma~\ref{prop:two-block-signed-pinching} (\ref{LSX3}) applied with
\[
P=[br_i]\, ,\qquad Q=[br_j]\, ,\qquad \text{and} \qquad D=\left[\frac{p}{2}(r_i+r_j)\right]\, ,
\]
yields $S(B_{ij})\le f(br_i)+f(br_j)$. Therefore,
\begin{align*}
S(\mathcal D_p\star\rho)
&\le\sum_i f(cr_i)+\sum_{i<j}\big(f(br_i)+f(br_j)\big)\\
&=\sum_i f(cr_i)+(d-1)\sum_i f(br_i)
&& \text{since each $i$ occurs in exactly $d-1$ pairs}\\
&=S(\rho)-c\log c-(d-1)b\log b\\
&\le S(\rho)+\log d
&& \text{since $(c,b,\ldots,b)$ is a probability vector}\\
&\le2\log d\, ,
\end{align*}
thus establishing the upper bound.

For the lower bound, the case $d=1$ is trivial and the case $d=2$ is Proposition~\ref{DPLQBT58}. Assume $d\ge3$ and $p<1$. Set $\alpha=p+2b$. Then $0<c\le\alpha<1$, and $q=p/\alpha\in[0,1]$. Let $\mathcal D_q^{(2)}$ denote the qubit depolarizing channel with parameter $q$. For each $i<j$, set $m_{ij}=r_i+r_j$ and
\[
Y_{ij}=[cr_i]\oplus B_{ij}\oplus[cr_j].
\]
If $m_{ij}>0$, define
\[
\rho_{ij}=m_{ij}^{-1}\left(r_i\dyad{0}{0}+r_j\dyad{1}{1}\right).
\]
A direct comparison of the blocks then yields
\[
Y_{ij}=\alpha m_{ij}(\mathcal D_q^{(2)}\star\rho_{ij})\, .
\]
If $m_{ij}=0$, then $r_i=r_j=0$ and hence $Y_{ij}=0$. For $m_{ij}>0$, since $\tr(\mathcal D_q^{(2)}\star\rho_{ij})=1$, Proposition~\ref{DPLQBT58} and~\eqref{ENTSCL54} give
\begin{align*}
S(Y_{ij})
&=\alpha m_{ij}S(\mathcal D_q^{(2)}\star\rho_{ij})-\alpha m_{ij}\log(\alpha m_{ij})\\
&\ge\alpha m_{ij}S(\rho_{ij})-\alpha m_{ij}\log(\alpha m_{ij})\\
&=\alpha(f(r_i)+f(r_j))-\alpha m_{ij}\log\alpha\, .
\end{align*}
For $m_{ij}=0$, the same inequality holds trivially. On the other hand,
\[
\sum_{i<j}S(Y_{ij})=(d-1)\sum_i f(cr_i)+\sum_{i<j}S(B_{ij})=S(\mathcal D_p\star\rho)+(d-2)\sum_i f(cr_i)\, ,
\]
while summing the preceding lower bound gives
\[
\sum_{i<j}S(Y_{ij})\ge\alpha\sum_{i<j}(f(r_i)+f(r_j))-\alpha\log\alpha\sum_{i<j}m_{ij}=\alpha(d-1)(S(\rho)-\log\alpha)\, .
\]
Since $\sum_i f(cr_i)=cS(\rho)-c\log c$ and $(d-1)\alpha-(d-2)c=1$,
\[
S(\mathcal D_p\star\rho)\ge S(\rho)+(d-2)c\log c-(d-1)\alpha\log\alpha\, .
\]
Finally, since the function $\phi(x)=-x\log x/(1-x)$ satisfies
\[
\phi'(x)=\frac{x-1-\log x}{(1-x)^2}\ge0
\]
for $0<x<1$, it follows that $\phi$ is increasing on $(0,1)$. Therefore,
\begin{align*}
(d-1)(-\alpha\log\alpha)
&=(d-1)(1-\alpha)\phi(\alpha)\\
&=(d-2)(1-c)\phi(\alpha) &&\text{since $(d-1)(1-\alpha)=(d-2)(1-c)$}\\
&\ge(d-2)(1-c)\phi(c)
&& \text{since $\alpha\ge c$}\\
&=(d-2)(-c\log c)\, ,
\end{align*}
thus $S(\mathcal D_p\star\rho)\ge S(\rho)$. The case $p=1$ follows from Proposition~\ref{UNTRXY71}, since in such a case $\mathcal D_1=\id_A$ is a unitary channel.
\end{proof}

We now consider bit-flip and phase-flip channels, which are central to the study of quantum error correction.

\begin{example}[Bit-flip channels]
Consider the qubit bit-flip channel 
\[
\mathcal B_p(X)=(1-p)X+p\sigma_xX\sigma_x,\qquad 0\leq p\leq1\, ,
\]
and let $\rho_a=a\dyad{0}{0}+(1-a)\dyad{1}{1}$ for $0\le a\le1$. The PDM entropy is symmetric about $p=1/2$, i.e.,
\[
S(\mathcal B_p\star\rho_a)=S(\mathcal B_{1-p}\star\rho_a)\, ,
\]
and satisfies $0\le S(\mathcal B_p\star\rho_a)\le\log2$ for all $a,p\in[0,1]$. The upper bound is attained for $a=1/2$ for every $p$, while the lower bound is attained when $a\in\{0,1\}$ and $p\in\{0,1\}$. A direct calculation yields the spectrum
\[
\sigma(\mathcal B_p\star\rho_a)
=
\left\{
\frac{1-p+d_1}{2},
\frac{1-p-d_1}{2},
\frac{p+d_2}{2},
\frac{p-d_2}{2}
\right\},
\]
where
\[
d_1=\sqrt{(1-p)^2(2a-1)^2+p^2},
\qquad
d_2=\sqrt{p^2(2a-1)^2+(1-p)^2}.
\]
Under $p\mapsto1-p$, the quantities $d_1$ and $d_2$ are interchanged, as are the two pairs of eigenvalues, which proves the symmetry above. A direct evaluation of the PDM entropy from the eigenvalues yields the stated bounds. In Figure~\ref{BFLPFIG72} we plot the PDM entropy for several values of $a$ as a function of $p\in[0,1]$.
\end{example}

\begin{figure}[htbp] 
\centering
\includegraphics[width=0.8\textwidth]{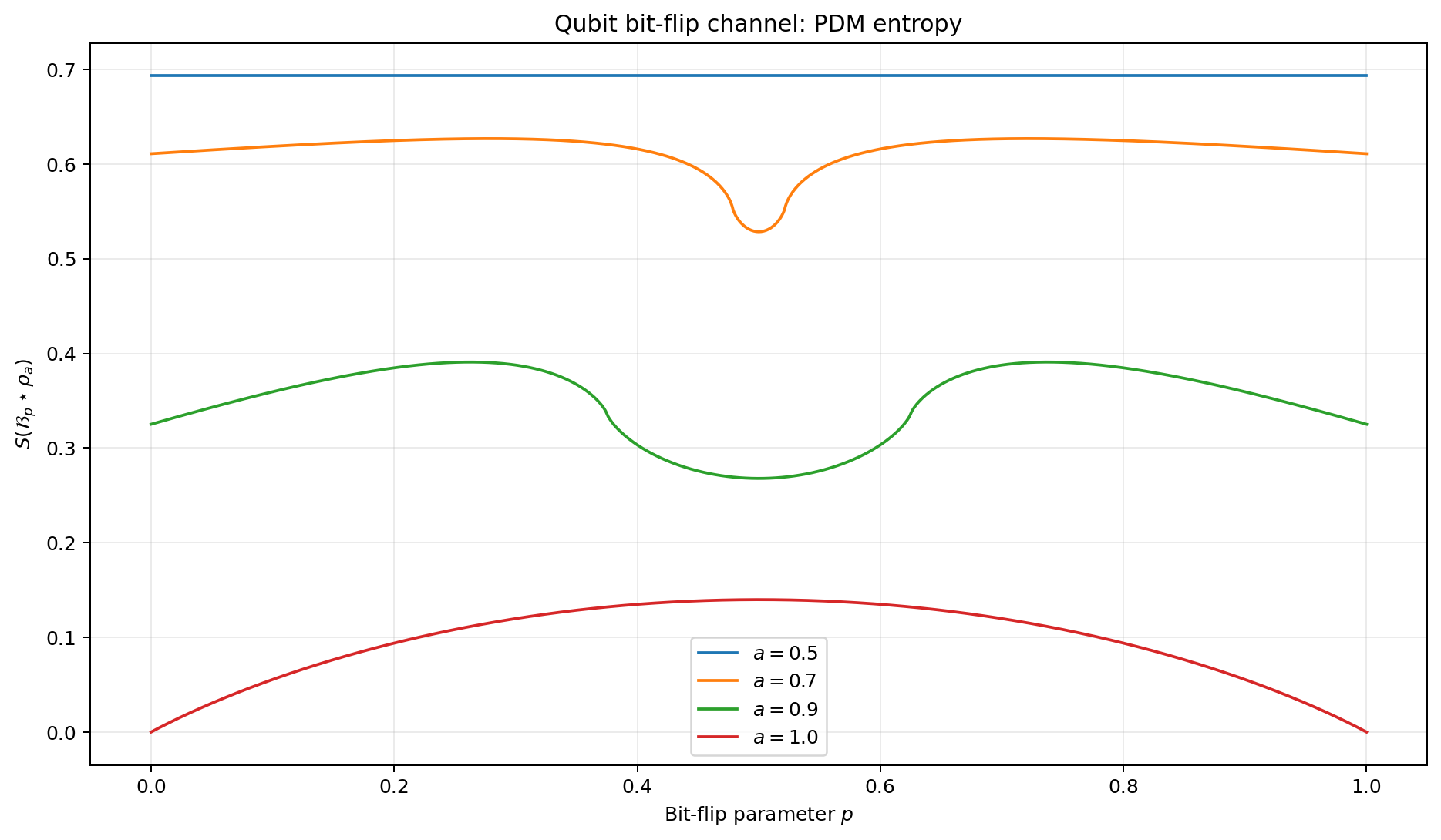}
\caption{The PDM entropy $S(\mathcal B_p\star\rho_a)$ for the qubit bit-flip channel, where $\rho_a=a\dyad{0}{0}+(1-a)\dyad{1}{1}$, for $a=0.5,0.7,0.9,1$. The entropy is symmetric about $p=1/2$ and remains between $0$ and $\log2$.}
\label{BFLPFIG72}
\end{figure}


\begin{example}[Phase-flip channels]
Consider the qubit phase-flip channel $\mathcal{P}_p:M_2(\C)\to M_2(\C)$ given by 
\[
\mathcal P_p(X)=(1-p)X+p\sigma_zX\sigma_z\qquad 0\leq p\leq 1\, , 
\]
and let $\rho_a=a\dyad{0}{0}+(1-a)\dyad{1}{1}$ for $0\le a\le1$. In this case the PDM entropy is independent of $p$ and satisfies
\[
S(\mathcal P_p\star\rho_a)=S(\rho_a)\, ,
\]
so that $0\le S(\mathcal P_p\star\rho_a)\le\log2$. Indeed, a direct calculation gives 
\[
\sigma(\mathcal P_p\star\rho_a)=\{a,1-a,(1-2p)/2,-(1-2p)/2\}\, ,
\]
thus the contributions of the last two eigenvalues to the PDM entropy cancel.
\end{example}

Finally we consider the case of amplitude-damping channels, which are typically utilized to model energy dissipation and decoherence. 

\begin{example}[Amplitude-damping channels]
Consider the qubit amplitude-damping channel $\mathcal A_\gamma(X)=K_0XK_0^\dagger+K_1XK_1^\dagger$ for $0\le\gamma\le1$, where $K_0=\dyad{0}{0}+\sqrt{1-\gamma}\dyad{1}{1}$ and $K_1=\sqrt{\gamma}\dyad{0}{1}$, and let $\rho_a=a\dyad{0}{0}+(1-a)\dyad{1}{1}$ for $0\le a\le1$. A direct calculation of the spectrum of $\mathcal A_\gamma\star\rho_a$ yields
\[
\sigma(\mathcal A_\gamma\star\rho_a)=\left\{a,\tilde{a}(1-\gamma),\frac{\tilde{a}\gamma+\sqrt{\tilde{a}^2\gamma^2+1-\gamma}}{2},\frac{\tilde{a}\gamma-\sqrt{\tilde{a}^2\gamma^2+1-\gamma}}{2}\right\}\, ,
\]
where $\tilde{a}=1-a$. Consequently, $S(\mathcal A_0\star\rho_a)=S(\mathcal A_1\star\rho_a)=S(\rho_a)$ and $S(\mathcal A_\gamma\star\mathds{1}_2/2)=\log2$ for every $\gamma\in[0,1]$. Theorem~\ref{thm:qubit-input-sot-mutual-information} also gives $S(\mathcal A_\gamma\star\rho_a)\le S(\rho_a)+S(\mathcal A_\gamma(\rho_a))\le2\log2$. In Figure~\ref{AMPDFIG83} we plot the PDM entropy for several values of $a$ as a function of the amplitude-damping parameter $\gamma\in [0,1]$.
\end{example}

\begin{figure}[htbp]
\centering
\includegraphics[width=0.8\textwidth]{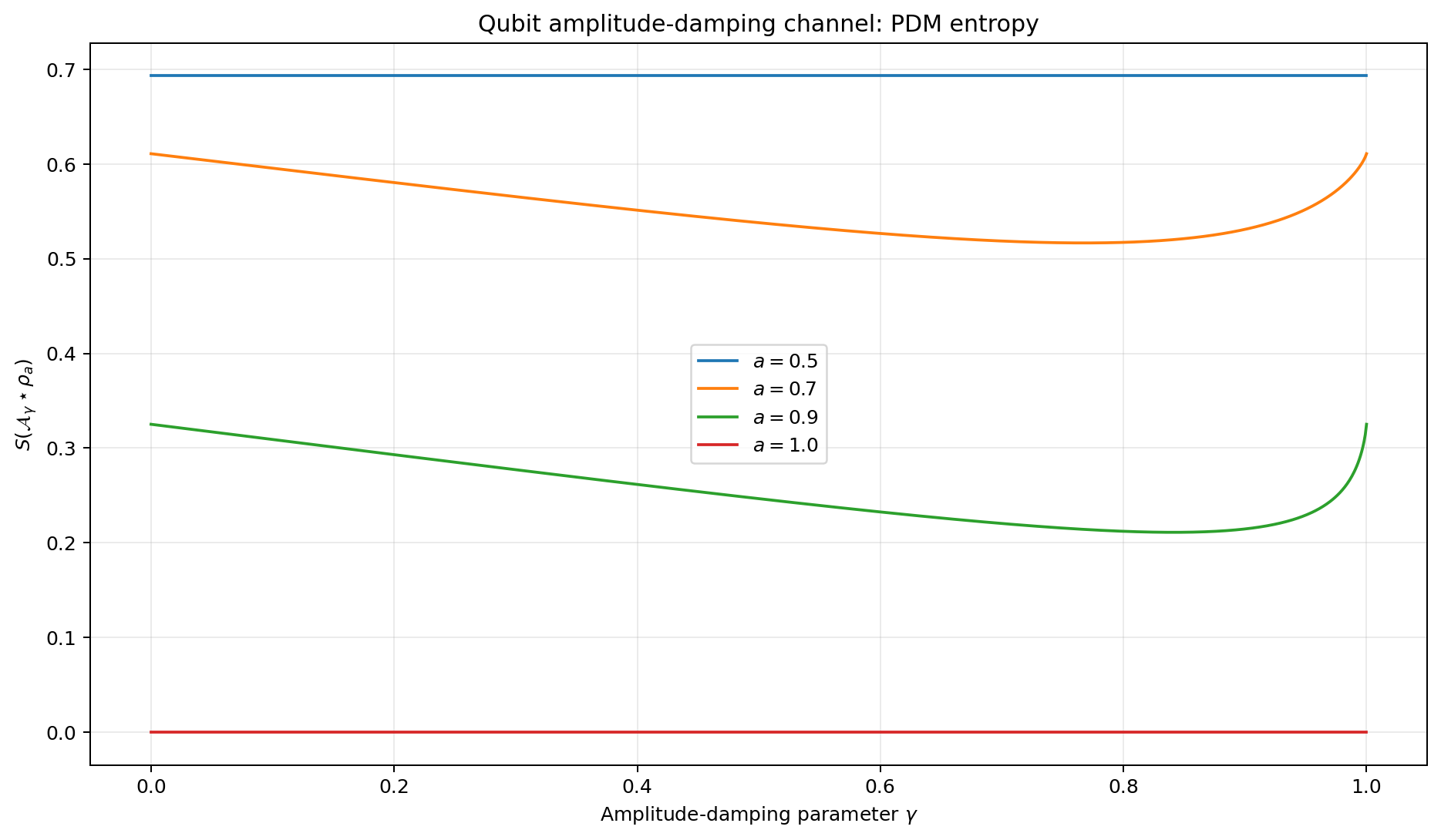}
\caption{The PDM entropy $S(\mathcal A_\gamma\star\rho_a)$ for the qubit amplitude-damping channel, where $\rho_a=a\dyad{0}{0}+(1-a)\dyad{1}{1}$, for $a=0.5,0.7,0.9,1$. The endpoint values satisfy $S(\mathcal A_0\star\rho_a)=S(\mathcal A_1\star\rho_a)=S(\rho_a)$, while $S(\mathcal A_\gamma\star\mathds{1}_2/2)=\log2$ for every $\gamma\in[0,1]$.}
\label{AMPDFIG83}
\end{figure}

\section{Concluding remarks}

In this work, we derived a unique extension of von~Neumann entropy to pseudo-density matrices from two simple assumptions: unitary invariance and strong additivity with respect to affine combinations within the interval $[-1,1]$. The relaxation of strong additivity over convex combinations to allow for affine combinations within the interval $[-1,1]$ was motivated by our proof that $[-1,1]$ always contains the eigenvalues of a pseudo-density matrix. While the entropy of a pseudo-density matrix also satisfies properties in common with von~Neumann entropy, such as additivity over tensor products and a Fannes-Audenaert type inequality, it is presently unknown if PDM entropy is subadditive for 2-time pseudo-density matrices. Although we proved subadditivity for 2-time pseudo-density matrices whose initial state is that of a single qubit, numerically generated examples suggest that subadditivity may fail in higher dimensions. However, we still do not know of a simple, physically motivated example where subadditivity fails.

As each eigenvalue $\lambda\in [-1,1]$ of a pseudo-density matrix contributes $-\lambda\log|\lambda|$ to its entropy, negative eigenvalues yield a negative contribution to the entropy, making it difficult to obtain explicit bounds on PDM entropy for general pseudo-density matrices. In particular, we do not even know if PDM entropy is non-negative in general when restricted to pseudo-density matrices. However, in the examples considered in this work (and in prior works), the PDM entropy behaves in a way that is physically intuitive, serving as an effective global measure of uncertainty which takes into account both the entropy of the initial state and the uncertainty induced by open system dynamics. In any case, an explicit, operational meaning of PDM entropy is still lacking, and remains a primary objective for future investigation. Having established a rigorous mathematical justification for its use, we naturally anticipate that a deeper understanding of PDM entropy will lead to new insights regarding spatiotemporal aspects of quantum information.

\emph{Acknowledgments.}---The authors acknowledge the use of Google Gemini to assist in the exploration of analytical methods and the preliminary structural organization of mathematical proofs. The authors take full responsibility for the logical derivation, rigorous verification, and final presentation of all mathematical results. J.F. is supported by the Hainan Provincial Natural Science Foundation of China under Grant No.~126MS0010.

\addcontentsline{toc}{section}{\numberline{}Bibliography}
\bibliographystyle{quantum}
\bibliography{references}

\Addresses
\end{document}